\documentclass{amsart}

\usepackage{amssymb, amscd,txfonts,stmaryrd,setspace,makecell}
\usepackage{enumerate}
\usepackage[usenames,dvipsnames]{xcolor}
\usepackage[bookmarks=true,colorlinks=true, linkcolor=MidnightBlue, citecolor=cyan]{hyperref}
\usepackage{lmodern}
\usepackage{graphicx,float}
\usepackage{tikz}
\usetikzlibrary{matrix,arrows}

\input xy
\xyoption{all} \xyoption{poly} \xyoption{knot}\xyoption{curve}
\usepackage{xypic,color}
\newcommand{\comment}[1]{}

\def\tn{\textnormal}

\def\mc{\mathcal}

\def\ZZ{{\mathbb Z}}

\def\NN{{\mathbb N}}
\def\SS{{\mathbb S}}

\def\Hom{\tn{Hom}}
\def\Mor{\tn{Mor}}

\def\Path{\tn{Path}}
\def\Ob{\tn{Ob}}

\def\SEL*{\tn{SEL*}}

\def\hsp{\hspace{.3in}}

\newcommand{\tin}[1]{\text{\tiny #1}}

\def\to{\rightarrow}
\def\from{\leftarrow}

\def\Down{\Downarrow}
\def\cross{\times}
\def\taking{\colon}

\def\too{\longrightarrow}

\def\ss{\subseteq}

\def\iso{\cong}

\def\|{{\;|\;}}
\def\m1{{-1}}

\def\wh{\widehat}

\def\ul{\underline}

\newcommand{\LMO}[1]{\stackrel{#1}{\bullet}}
\newcommand{\LTO}[1]{\stackrel{\tt{#1}}{\bullet}}
\newcommand{\LA}[2]{\ar[#1]^-{\tn {#2}}}
\newcommand{\LAL}[2]{\ar[#1]_-{\tn {#2}}}
\newcommand{\obox}[3]{\stackrel{#1}{\fbox{\parbox{#2}{#3}}}}

\def\ullimit{\ar@{}[rd]|(.3)*+{\lrcorner}}
\def\urlimit{\ar@{}[ld]|(.3)*+{\llcorner}}
\def\lllimit{\ar@{}[ru]|(.3)*+{\urcorner}}
\def\lrlimit{\ar@{}[lu]|(.3)*+{\ulcorner}}
\def\ulhlimit{\ar@{}[rd]|(.3)*+{\diamond}}
\def\urhlimit{\ar@{}[ld]|(.3)*+{\diamond}}
\def\llhlimit{\ar@{}[ru]|(.3)*+{\diamond}}
\def\lrhlimit{\ar@{}[lu]|(.3)*+{\diamond}}
\newcommand{\clabel}[1]{\ar@{}[rd]|(.5)*+{#1}}
\newcommand{\TriRight}[7]{\xymatrix{#1\ar[dr]_{#2}\ar[rr]^{#3}&&#4\ar[dl]^{#5}\\&#6\ar@{}[u] |{\Longrightarrow}\ar@{}[u]|>>>>{#7}}}
\newcommand{\TriLeft}[7]{\xymatrix{#1\ar[dr]_{#2}\ar[rr]^{#3}&&#4\ar[dl]^{#5}\\&#6\ar@{}[u] |{\Longleftarrow}\ar@{}[u]|>>>>{#7}}}
\newcommand{\TriIso}[7]{\xymatrix{#1\ar[dr]_{#2}\ar[rr]^{#3}&&#4\ar[dl]^{#5}\\&#6\ar@{}[u] |{\Longleftrightarrow}\ar@{}[u]|>>>>{#7}}}

\newcommand{\arr}[1]{\ar@<.5ex>[#1]\ar@<-.5ex>[#1]}
\newcommand{\arrr}[1]{\ar@<.7ex>[#1]\ar@<0ex>[#1]\ar@<-.7ex>[#1]}
\newcommand{\arrrr}[1]{\ar@<.9ex>[#1]\ar@<.3ex>[#1]\ar@<-.3ex>[#1]\ar@<-.9ex>[#1]}
\newcommand{\arrrrr}[1]{\ar@<1ex>[#1]\ar@<.5ex>[#1]\ar[#1]\ar@<-.5ex>[#1]\ar@<-1ex>[#1]}

\newcommand{\To}[1]{\xrightarrow{#1}}
\newcommand{\Too}[1]{\xrightarrow{\ \ #1\ \ }}

\newcommand{\Froom}[1]{\xleftarrow{\ \ #1\ \ }}

\newcommand{\Adjoint}[4]{\xymatrix@1{#2 \ar@<.5ex>[r]^-{#1} & #3 \ar@<.5ex>[l]^-{#4}}}

\def\id{\tn{id}}
\def\Top{{\bf Top}}
\def\Cat{{\bf Cat}}
\def\cpo{{\bf cpo}}

\def\Str{\tn{Str}}

\def\Inst{{\bf Inst}}
\def\Type{{\bf Type}}
\def\Set{{\bf Set}}

\def\set{{\text \textendash}{\bf Set}}
\def\inst{{{\text \textendash}\bf \Inst}}

\def\bhline{\Xhline{2\arrayrulewidth}}
\def\bbhline{\Xhline{2.5\arrayrulewidth}}

\def\mcB{\mc{B}}
\def\mcC{\mc{C}}
\def\mcD{\mc{D}}
\def\mcE{\mc{E}}
\def\mcF{\mc{F}}
\def\mcG{\mc{G}}

\def\mcS{\mc{S}}

\newtheorem{theorem}{Theorem}[subsection]
\newtheorem{lemma}[theorem]{Lemma}
\newtheorem{proposition}[theorem]{Proposition}

\theoremstyle{remark}
\newtheorem{remark}[theorem]{Remark}
\newtheorem{example}[theorem]{Example}

\newtheorem{question}[theorem]{Question}
\newtheorem{guess}[theorem]{Guess}

\theoremstyle{definition}
\newtheorem{definition}[theorem]{Definition}
\newtheorem{notation}[theorem]{Notation}

\newtheorem{construction}[theorem]{Construction}

\def\Sch{{\bf Sch}}
\def\Fin{{\bf Fin}}

\newcommand{\labelDisp}[2]{\begin{align}\label{#1}\text{#2}\end{align}}

\newcommand{\mainCatLarge}[1]{ 
	\stackrel{#1}{
		\parbox{2.85in}{\small\fbox{\parbox{2.85in}{\underline{Mgr;Dpt$ \simeq$ Dpt}\hsp  \underline{Secr;Dpt$ \simeq\id_{\tn{Department}}$}\\\\
			\xymatrix@=4pt{&\LTO{Employee}\ar@<.5ex>[rrrrr]^{\tn{Dpt}}\ar@(l,u)[]^{\tn{Mgr}}\ar[dddl]_{\tn{First}}\ar[dddr]^{\tn{Last}}&&&&&\LTO{Department}\ar@<.5ex>[lllll]^{\tn{Secr}}\ar[ddd]^{\tn{Name}}\\\\\\\LTO{String1}&&\LTO{String2}&~&~&~&\LTO{String3}
			}
		}}}
	}
}

\CompileMatrices
\begin{document}

\title{Functorial data migration}

\author{David I. Spivak}

\address{Department of Mathematics, Massachusetts Institute of Technology, Cambridge MA 02139}

\email{dspivak@math.mit.edu}

\thanks{This project was supported by ONR grant N000141010841.}

\begin{abstract}

In this paper we present a simple database definition language: that of categories and functors. A database schema is a small category and an instance is a set-valued functor on it. We show that morphisms of schemas induce three ``data migration functors", which translate instances from one schema to the other in canonical ways. These functors parameterize projections, unions, and joins over all tables simultaneously and can be used in place of conjunctive and disjunctive queries. We also show how to connect a database and a functional programming language by introducing a functorial connection between the schema and the category of types for that language. We begin the paper with a multitude of examples to motivate the definitions, and near the end we provide a dictionary whereby one can translate database concepts into category-theoretic concepts and vice-versa.

\end{abstract}

\maketitle

\tableofcontents

\section{Introduction}\label{sec:intro}

This paper has two main goals. The first goal is to present a straightforward category-theoretic model of databases under which every theorem about small categories becomes a theorem about databases. To do so, we will present a category $\Sch$ of database schemas, which has three important features: \begin{itemize}\item the category $\Sch$ is equivalent to $\Cat$, the category of small categories, \item the category $\Sch$ is a faithful model for real-life database schemas, and\item the category $\Sch$ serves as a foundation upon which high-level database concepts rest easily and harmoniously.\end{itemize} 

The second goal is to apply this category-theoretic formulation to provide new data migration functors, so that for any translation of schemas $F\taking\mcC\to\mcD$, one can transport instances on the source schema $\mcC$ to instances on the target schema $\mcD$ and vice versa, with provable ``round-trip" properties. For example, homomorphisms of instances are preserved under all migration functors. While these migration functors do not appear to have been discussed in database literature, their analogues are well-known in modern programming languages theory, e.g. the theory of dependent types (\cite{Mar}), and polynomial data types (\cite{Kan}). This is part of a deeper connection between database schemas and {\em kinds} (structured collections of types, see \cite{SH}, \cite{Ham}) in programming languages. See also Section \ref{sec:Grothendieck}.

An increasing number of researchers in an increasing variety of disciplines are finding that categories and functors offer high-quality models, which simplify and unify their respective fields.\footnote{Aside from mathematics, in which category-theoretic language and theorems are indispensable in modern algebra, geometry, and topology, category theory has been successful in: programming language theory \cite{Mog},\cite{Pie}; physics \cite{BS},\cite{DI}; materials science \cite{SGWB},\cite{GSB}; and biology \cite{Ros},\cite{EV}.} The quality of a model should be judged by its efficiency as a proxy or interface---that is, by the ease with which an expert can work with only an understanding of the model, and in so doing successfully operate the thing itself. Our goal in this paper is to provide a high-quality model of databases.

Other category-theoretic models of databases have been presented in the past (\cite{JoM},\cite{BCW},\cite{JRW},\cite{IP},\cite{PS},\cite{DC},\cite{Tui}). Almost all of them used the more expressive notion of sketch where we have used categories. The additional expressivity came at a cost that can be cast in terms of our two goals for this paper. First, the previous models were more complex and this may have created a barrier to wide-spread understanding and adoption. Second, morphisms of sketches do not generally induce the sorts of data migration functors that morphisms of categories do. 

It is our hope that the present model is simple enough that anyone who has an elementary understanding of categories (i.e. who knows the definition of category, functor, and natural transformation) will, without too much difficulty, be able to understand the basic idea of our formulation: database schemas as categories, database instances as functors. Moreover, we will provide a dictionary (see Section \ref{sec:dictionary}, Table \ref{table:dictionary}) whereby the main results and definitions in this paper will simply correspond to results or definitions of standard category theory; this way, the reader can rely on tried and true sources to explain the more technical ideas presented here. Moreover, one may hope to leverage existing mathematical theory to their own database issues through this connection.

Before outlining the plan of the paper in Section \ref{sec:plan} we will give a short introduction to the fundamental idea on which the paper rests, and provide a corresponding ``categorical normal form" for databases, in Section \ref{sec:CNF}.

\subsection{Categorical normal form}\label{sec:CNF}

A database schema may contain hundreds of tables and foreign keys. Each foreign key links one table to another, and each sequence of foreign keys $T_1\to T_2\to \cdots\to T_n$ results in a function $f$ from the set of records in $T_1$ to the set of records in $T_n$. It is common that two different foreign key paths, both connecting table $T_1$ to table $T_n$, may exist; and they may or may not define the same mapping on the level of records. For example, an operational ``landline" phone is assigned a phone number whose area code corresponds to the region in which the physical phone is located. Thus we have two paths $OLP\to R$:
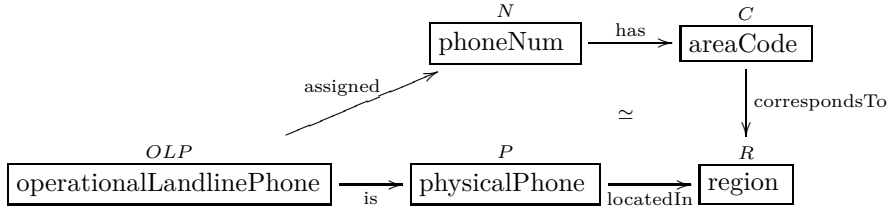
\begin{align}\label{dia:phone paths}\xymatrix{&\obox{N}{.7in}{phoneNum}\LA{r}{has}&\obox{C}{.6in}{areaCode}\ar@{}[dl]|{\simeq}\LA{d}{correspondsTo}\\\obox{OLP}{1.6in}{operationalLandlinePhone}\LA{ru}{assigned}\LAL{r}{is}&\obox{P}{.9in}{physicalPhone}\LAL{r}{locatedIn}&\obox{R}{.4in}{region}}\end{align} and the data architect for this schema knows whether or not these two paths should always produce the same mapping. In (\ref{dia:phone paths}), the two paths $OLP\to R$ do produce the same mapping, and the $\simeq$ sign is intended to record that fact. For contrast, if we replace operationalLandlinePhone with operationalMobilePhone (OMP), the two paths $OMP\to R$ would not produce the same mapping, because a cellphone need not be currently located in the region indicated by its area code. Thus we would get a similar but different diagram,
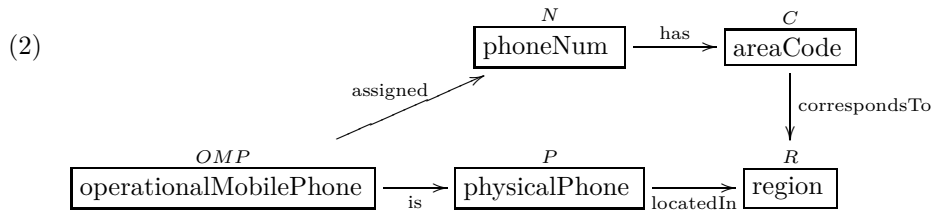
\begin{align}\label{dia:phone paths2}\xymatrix{&\obox{N}{.7in}{phoneNum}\LA{r}{has}&\obox{C}{.6in}{areaCode}\LA{d}{correspondsTo}\\\obox{OMP}{1.5in}{operationalMobilePhone}\LA{ru}{assigned}\LAL{r}{is}&\obox{P}{.9in}{physicalPhone}\LAL{r}{locatedIn}&\obox{R}{.4in}{region}}\end{align} 

We are emphasizing here that the notion of {\em path equivalence} is an important and naturally-arising integrity constraint, which provides a crucial clue into the intended semantics of the schema. Its enforcement is often left to the application layer, but it should actually be included as part of the schema (\cite{JD}). Including path equivalence information in the database schema has three main advantages: \begin{itemize}\item it permits the inclusion of ``hot" query columns without redundancy, \item it provides an important check for creating schema mappings, and \item it promotes healthy schema evolution.\end{itemize}

A {\em category} (in the mathematical sense) is roughly a graph with one additional bit of expressive power: the ability to declare two paths equivalent.  We now have the desired connection between database schemas and categories: Tables in a schema are specified by vertices (or as we have drawn them in Diagram (\ref{dia:phone paths}), by boxes); columns are specified by arrows; and functional equivalence of foreign key paths are specified by the category-theoretic notion of path equivalence (indicated by the $\simeq$ symbol). The categorical definition of schema will be presented rigorously in Section \ref{sec:definitions}.

The above collection of ideas leads us to the following normal form for databases.

\begin{definition}\label{def:CNF}

A database is in {\em categorical normal form} if 
\begin{itemize}
\item every table $t$ has a single primary key column $\text{ID}_t$, chosen at the outset.  The cells in this column are called the {\em row-ids} of $t$;
\item for every column $c$ of a table $t$, there exists some {\em target table} $t'$ such that the value in each cell of column $c$ refers to some row-id of $t'$. We denote this relationship by $$c\taking t\to t';$$ 
\item in particular, if some column $d$ of $t$ consists of ``pure data" (such as strings or integers), then its target table $t'$ is simply a 1-column table (i.e. a controlled vocabulary, containing at least the active domain of column $d$), and we still write $d\taking t\to t'$; and 
\item when there are two paths $p,q$ through the database from table $t$ to table $u$ (denoted $p\taking t\to u,\;q\taking t\to u$) and it is known to the schema designer that $p$ and $q$ should correspond to the same mapping of row-ids, then this path equivalence must be declared as part of the schema. We denote this path equivalence by $$p\simeq q.$$\end{itemize}

\end{definition}

\subsection{Plan of the paper}\label{sec:plan}

We begin in Section \ref{sec:examples} by giving a variety of examples, which illustrate the virtues of the category-theoretic approach. These include conceptual clarity, simplified data migration, updatable views, interoperability with RDF data, and a close connection between data and program.  In Section \ref{sec:definitions} we will give the precise definitions of categorical schemas and translations, and show that the category thereof, denoted $\Sch$, is equivalent to $\Cat$, the category of small categories. In this section we will also define the category of database instances on a given schema. In Section \ref{sec:data migration} we will define the data migration functors associated to a translation and begin to wrap up our tour by returning to the examples from Section \ref{sec:examples} for a more in-depth treatment. We finish this work in Section \ref{sec:typing}, where we discuss data types and filtering. 

\subsection{Terminology and notation}\label{sec:notation}

One obstacle to writing this paper is a certain overlap in terminology between databases and categories: the word ``object" is commonly used in both contexts. While object databases are interesting and perhaps relevant to some of the ideas presented here, we will not discuss them at all, hence keeping the namespace clear for categorical terminology. {\bf Unless otherwise specified, the word {\em object} will always be intended in the category-theoretic sense.} 

Say we have maps $A\To{f}B\To{g}C$; we may denote their composite $A\to C$ in one of two ways, depending on what seems more readable. The first is called ``diagrammatic order" and is written as $f;g$. The second is called ``classical order" and is written as $g\circ f$. We may sometimes choose not to write a symbol between $f$ and $g$, and in that case we use diagrammatic order $fg\taking A\to C$.

\subsection{Acknowledgments}

The present paper is a total revamping of \cite{Sp1}, which, in an attempt to accommodate two disjoint communities (mathematicians and database experts), ended up as a sprawling and somewhat incoherent document. My thanks go to the referees and to Bob Harper for pointing me toward a streamlined presentation. I would also like to thank Peter Gates, Dave Balaban, John Launchbury, and Greg Morrisett for many useful conversations. I appreciate very much the encouragement of Jack Morava, especially in the form his vision (\cite{Mor}) of how ideas from this paper may be useful for exposing patterns in pure mathematics. Special thanks also go to Scott Morrison for coding many of the ideas from this paper into working form, available online for demonstration or open-source participation at \url{http://categoricaldata.net/}.

\section{Virtues by Example}\label{sec:examples}

In what follows we illustrate the merits of the category-theoretic approach by way of several examples. One thing to note is that each of these features flows naturally from our compact mathematical definitions of schemas and translations. These definitions will be given in Section \ref{sec:definitions}. 

\subsection{Conceptual clarity}

In a categorical schemas (Definition \ref{def:categorical schema}), every table is a vertex and every column is an arrow. An arrow $\LTO{T}\Too{c}\LTO{U}$ represents a column of table {\tt T}, with target table {\tt U}, i.e. a foreign key constraint declaring that each cell in column $c$ refers to a row-id in table {\tt U}. We draw a box around our system of vertices and arrows, and the result is a categorical representation of the schema\footnote{A system of vertices and arrows of this sort is called a graph. A graph can be considered as a kind of category (a so-called {\em free category}) in which no path equivalences have been declared.}.

\begin{example}\label{ex:two fact schema}

The following picture represents a schema $S$ with six tables, two of which are multi-column ``fact tables" and four of which are 1-column ``leaf" tables. \begin{align}\label{dia:two facts}\mcC:=\parbox{1.5in}{\fbox
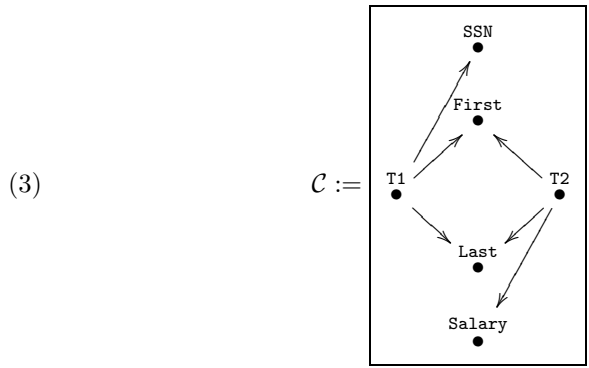
{\xymatrix@=10pt{&\LTO{SSN}\\&\LTO{First}\\\LTO{T1}\ar[uur]\ar[ur]\ar[dr]&&\LTO{T2}\ar[ul]\ar[dl]\ar[ddl]\\&\LTO{Last}\\&\LTO{Salary}}}}\end{align} The fact table ${\tt T1}$ has three columns (pointing to {\tt SSN}, {\tt First}, {\tt Last}), in addition to its ID column; the fact table ${\tt T2}$ also has three non-ID columns (pointing to {\tt First}, {\tt Last}, {\tt Salary}). The leaf tables, {\tt SSN}, {\tt First}, {\tt Last}, and {\tt Salary} have only ID columns, as is seen by the fact that no arrows emanate from them.

As a set of tables, an instance on schema $\mcC$ may look something like this: 
\labelDisp{dia:facts}{
\begin{tabular}{| l || l | l | l |}\bhline\multicolumn{4}{| c |}{{\tt T1}}\\\bhline {\bf ID}&{\bf SSN}&{\bf First}&{\bf Last}\\\bbhline T1-001&115-234&Bob&Smith\\\hline T1-002&122-988&Sue&Smith\\\hline T1-003&198-877&Alice&Jones\\\bhline
\end{tabular}
\hsp
\begin{tabular}{| l || l | l | l |}\bhline\multicolumn{4}{| c |}{{\tt T2}}\\\bhline {\bf ID}&{\bf First}&{\bf Last}&{\bf Salary}\\\bbhline T2-A101&Alice&Jones&\$100\\\hline T2-A102&Sam&Miller&\$150\\\hline T2-A104&Sue&Smith&\$300\\\hline T2-A110&Carl&Pratt&\$200 \\\bhline
\end{tabular}
}
\labelDisp{dia:leaves}{
\begin{tabular}{| l ||}\bhline\multicolumn{1}{| c |}{{\tt SSN}}\\\bhline {\bf ID}\\\bbhline 115-234\\\hline 118-334\\\hline 122-988\\\hline 198-877\\\hline 342-164\\\bhline
\end{tabular}
\hsp\hsp
\begin{tabular}{| l ||}\bhline\multicolumn{1}{| c |}{{\tt First}}\\\bhline {\bf ID}\\\bbhline Adam\\\hline Alice\\\hline Bob\\\hline Carl\\\hline Sam\\\hline Sue\\\bhline
\end{tabular}
\hsp\hsp
\begin{tabular}{| l ||}\bhline\multicolumn{1}{| c |}{{\tt Last}}\\\bhline {\bf ID}\\\bbhline Jones\\\hline Miller\\\hline Pratt\\\hline Richards\\\hline Smith\\\bhline
\end{tabular}
\hsp\hsp
\begin{tabular}{| l ||}\bhline\multicolumn{1}{| c |}{{\tt Salary}}\\\bhline {\bf ID}\\\bbhline \$100\\\hline \$150\\\hline \$200\\\hline \$250\\\hline \$300\\\bhline
\end{tabular}
}

The thing to recognize here is that each column header $c$ of {\tt T1} (respectively, of {\tt T2}) points to some target table, in such a way that every cell in column $c$ refers to a row-id in that target table.\footnote{In the case that $c$ is the ID column of {\tt T1}, the target table to which $c$ points is {\tt T1}, and each cell in column $c$ is a row-id in {\tt T1} that refers to itself.} The leaf tables serve as controlled vocabularies for the fact tables.

\end{example}

\begin{notation}\label{notation:leaf tables}

In Example \ref{ex:two fact schema}, we wrote out all four leaf tables (Display \ref{dia:leaves}). In the future we will not generally write out leaf tables for space reasons. In fact, any table that does not add explanatory power to a given example may be left out from our displays.

\end{notation}

\begin{example}\label{ex:famous schema}

In this example we present a schema $\mcC$ that includes path equivalences, and hence takes advantage of the full expressivity of categories. We imagine a company with employees and departments; every employee is in a department, every employee has a manager employee, and every department has a secretary employee. Using path equivalences, we enforce the following facts:\begin{itemize} \item the manager of an employee is in the same department as that employee, and \item the secretary of a department is in that department.\end{itemize} These facts are recorded as equations at the top of the following diagram:

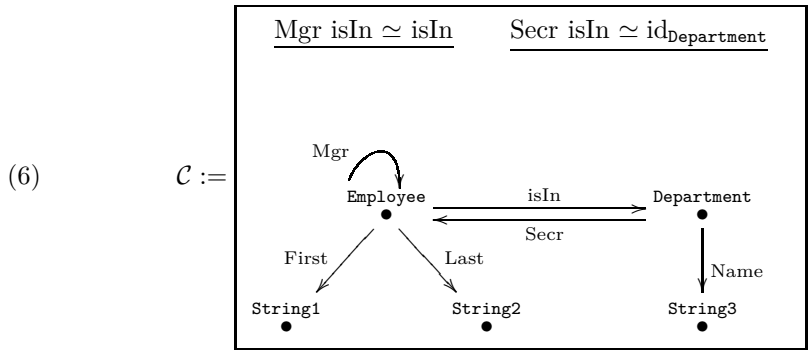
\begin{align}\label{dia:employee}\mcC:=\parbox{2.9in}{\fbox{\parbox{2.9in}{\begin{center}\underline{Mgr isIn $\simeq$ isIn}\hsp  \underline{Secr isIn $\simeq \id_{{\tt Department}}$}\end{center}~\\\\\\
			\xymatrix@=4pt{&\LTO{Employee}\ar@<.5ex>[rrrrr]^{\tn{isIn}}\ar@(l,u)[]+<5pt,10pt>^{\tn{Mgr}}\ar[dddl]_{\tn{First}}\ar[dddr]^{\tn{Last}}&&&&&\LTO{Department}\ar@<.5ex>[lllll]^{\tn{Secr}}\ar[ddd]^{\tn{Name}}\\\\\\\LTO{String1}&&\LTO{String2}&~&~&~&\LTO{String3}
			}
		}}}\vspace{1in}\end{align}

As a set of tables, an instance on $\mcC$ may look something like this: 

\begin{align*}&\begin{tabular}{| l || l | l | l | l |}\bhline\multicolumn{5}{| c |}{{\tt Employee}}\\\bhline {\bf ID}&{\bf First}&{\bf Last}&{\bf Mgr}&{\bf isIn}\\\bbhline 101&David&Hilbert&103&q10\\\hline 102&Bertrand&Russell&102&x02\\\hline 103&Alan&Turing&103&q10\\\bhline\end{tabular}&\begin{tabular}{| l || l | l |}\bhline\multicolumn{3}{| c |}{{\tt Department}}\\\bhline {\bf ID}&{\bf Name}&{\bf Secr}\\\bbhline q10&Sales&101\\\hline x02&Production&102\\\bhline\end{tabular}\end{align*}
 
It is no coincidence that there are a total of six non-ID columns in the two tables and a total of six arrows in the schema $\mcC$. The equations can be checked on these cells; for example we can check the first equation on row-id 101: \begin{center}101.Manager.isIn = 103.isIn = q10 \hsp and \hsp 101.isIn = q10.\end{center}

An instance $I$ on $\mcC$ is only valid if the two equations hold for every row in $I$.

\end{example}

\subsection{Simplified data migration}\label{sec:simplified data migration}

A translation $F\taking\mcC\to\mcD$ of schemas (Definition \ref{def:translation}) is a mapping that takes vertices in $\mcC$ to vertices in $\mcD$ and arrows in $\mcC$ to {\em paths} in $\mcD$; in so doing, it must respect arrow sources, arrow targets, and path equivalences. We will be using the following translation $F\taking\mcC\to\mcD$ throughout Section \ref{sec:simplified data migration}.

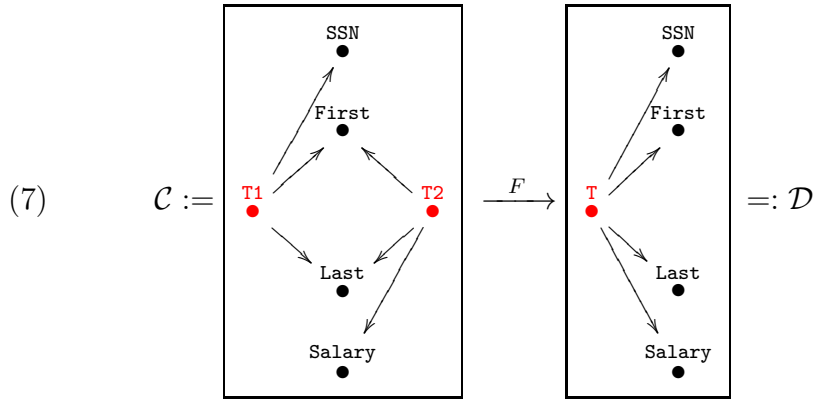
\Large\begin{align}\label{dia:translation}\mcC:=\parbox{1.3in}{\fbox{\xymatrix@=10pt{&\LTO{SSN}\\&\LTO{First}\\\color{red}{\LTO{T1}}\ar[uur]\ar[ur]\ar[dr]&&\color{red}{\LTO{T2}}\ar[ul]\ar[dl]\ar[ddl]\\&\LTO{Last}\\&\LTO{Salary}}}}\Too{F}\parbox{.9in}{\fbox{\xymatrix@=10pt{&\LTO{SSN}\\&\LTO{First}\\\color{red}{\LTO{T}}\ar[uur]\ar[ur]\ar[dr]\ar[ddr]\\&\LTO{Last}\\&\LTO{Salary}}}}=:\mcD\end{align}\normalsize

The mapping $F$ is drawn as suggestively as possible. In the future, we will rely on this ``power of suggestion" to indicate the translations, but this time we will be explicit. Each of the four leaf vertices, {\tt SSN}, {\tt First}, {\tt Last}, and {\tt Salary} in $\mcC$ is mapped to the vertex in $\mcD$ of the same label. The two other vertices in $\mcC$, namely {\tt T1} and {\tt T2}, are mapped to vertex {\tt T} in $\mcD$. Since translations must respect arrow sources and targets, there is no additional choice about where the arrows in $\mcC$ are sent; for example $F$ sends the arrow ${\tt T1}\to{\tt First}$ in $\mcC$  to the arrow ${\tt T}\to{\tt First}$ in $\mcD$. 

We have now specified a translation $F$ from schema $\mcC$ to schema $\mcD$, pictured in Diagram \ref{dia:translation}. Springing forth from this translation are three data migration functors, which we will discuss in turn in Examples \ref{ex:pullback}, \ref{ex:right pushforward}, and \ref{ex:left pushforward}. They migrate instance data on $\mcD$ to instance data on $\mcC$ and vice versa. Instances on $\mcC$ (respectively on $\mcD)$ form a category, which we denote by $\mcC\inst$ (respectively by $\mcD\inst$); they are defined in Definition \ref{def:category of instances}. We have the following chart, the jargon of which will be introduced shortly: \begin{center}\begin{tabular}{| c | c | c | c |}\bhline\multicolumn{4}{| c |}{Data migration functors induced by a translation $F\taking\mcC\to\mcD$}\\\bhline {\bf Name}&{\bf Symbol}&{\bf Long symbol}&{\bf Idea of definition}\\\bbhline Pullback & $\Delta_F$ & $\Delta_F\taking\mcD\inst\to\mcC\inst$&Composition with $F$\\\hline Right Pushforward&$\Pi_F$&$\Pi_F\taking\mcC\inst\to\mcD\inst$&Right adjoint to $\Delta_F$\\\hline Left Pushforward&$\Sigma_F$&$\Sigma_F\taking\mcC\inst\to\mcD\inst$&Left adjoint to $\Delta_F$\\\bhline\end{tabular}\end{center} Everything in the chart will be defined in Section \ref{sec:data migration}. For now we give three examples. In each, we will be starting with the translation $F\taking\mcC\to\mcD$, given above in Diagram (\ref{dia:translation}).

\begin{example}[Pullback]\label{ex:pullback}

Let $F\taking\mcC\to\mcD$ be the translation given in Diagram (\ref{dia:translation}). In this example, we explore the data migration functor $\Delta_F\taking\mcD\inst\to\mcC\inst$\footnote{We have not defined $\Delta_F$ yet; this will be done in Section \ref{sec:pullback}.} by applying it to a $\mcD$-instance $J$. Note that even though our translation $F$ points forwards (from $\mcC$ to $\mcD$), our migration functor $\Delta_F$ points ``backwards" (from $\mcD$-instances to $\mcC$-instances). We will see why it works that way, but first we bring the discussion down to earth by working with a particular $\mcD$-instance.

Consider the instance $J$, on schema $\mcD$, defined by the table 
\labelDisp{dia:tableT}{
J:=\;\begin{tabular}{| l || l | l | l | l |}\bhline\multicolumn{5}{| c |}{{\tt T}}\\\bhline {\bf ID}&{\bf SSN}&{\bf First}&{\bf Last}&{\bf Salary}\\\bbhline XF667&115-234&Bob&Smith&\$250\\\hline XF891&122-988&Sue&Smith&\$300\\\hline XF221&198-877&Alice&Jones&\$100\\\bhline
\end{tabular}
}
and having the four leaf tables from Example \ref{ex:two fact schema}, Display (\ref{dia:leaves}). Pulling back along the translation $F$, we are supposed to get an instance $\Delta_F(J)$ on schema $\mcC$, which we must describe. But the description is easy: $\Delta_F(J)$ splits up the columns of table ${\tt T}$ according to the translation $F$. The four leaf tables will be exactly the same as above (i.e. the same as in Example \ref{ex:two fact schema} (\ref{dia:leaves})), and the two fact tables will be something like\footnote{There may be choice in the naming convention for row-ids.}\labelDisp{dia:pullback facts}{
\begin{tabular}{| l || l | l | l |}\bhline\multicolumn{4}{| c |}{{\tt T1}}\\\bhline {\bf ID}&{\bf SSN}&{\bf First}&{\bf Last}\\\bbhline XF667T1&115-234&Bob&Smith\\\hline XF891T1&122-988&Sue&Smith\\\hline XF221T1&198-877&Alice&Jones\\\bhline
\end{tabular}
\hsp
\begin{tabular}{| l || l | l | l |}\bhline\multicolumn{4}{| c |}{{\tt T2}}\\\bhline {\bf ID}&{\bf First}&{\bf Last}&{\bf Salary}\\\bbhline A21&Alice&Jones&\$100\\\hline A67&Bob&Smith&\$250\\\hline A91&Sue&Smith&\$300\\\hline 
\end{tabular}
}

The fact that {\tt T1} and {\tt T2} are simply projections of {\tt T} is a result of our choice of translation $F$. 

\end{example}

\begin{remark}

We have seen that {\bf the pullback functor $\Delta_F$}, which arises naturally for any translation $F$ between schemas, {\bf automatically produces projections}. 

\end{remark}

In the next two examples, we will explore the right and left pushforward migration functors induced by the translation $F\taking\mcC\to\mcD$ given in Diagram (\ref{dia:translation}). These functors, denoted $\Pi_F$ and $\Sigma_F$, send $\mcC$-instances to $\mcD$-instances. Thus we start with the instance $I$ (which was presented in Example \ref{ex:two fact schema}) and explain its pushforwards $\Pi_F(I)$ and $\Sigma_F(I)$ below in Examples \ref{ex:right pushforward} and \ref{ex:left pushforward}, respectively. 
%

\begin{example}[Right Pushforward]\label{ex:right pushforward}

Let $F\taking\mcC\to\mcD$ be the translation given in Diagram (\ref{dia:translation}). In this example, we explore the data migration functor $\Pi_F\taking\mcC\inst\to\mcD\inst$\footnote{We have not defined $\Pi_F$ yet; this will be done in Section \ref{sec:right pushforward}.} by applying it to the $\mcC$-instance $I$ shown in Displays (\ref{dia:facts}) and (\ref{dia:leaves}). Note that our migration functor $\Pi_F$ points in the same direction as $F$: it takes $\mcC$-instances to $\mcD$-instances. We now describe the $\mcD$-instance $\Pi_F(I)$, which has four leaf tables $\Pi_F(I)({\tt SSN})$, etc., and one fact table $\Pi_F(I)({\tt T})$.

The four leaf tables of $\Pi_F(I)$ will be as in Display (\ref{dia:leaves}). The fact table of $\Pi_F(I)$ will be the join of {\tt T1} and {\tt T2}: 
\begin{center}

\begin{tabular}{| l || l | l | l | l |}\bhline\multicolumn{5}{| c |}{{\tt T}}\\\bhline {\bf ID}&{\bf SSN}&{\bf First}&{\bf Last}&{\bf Salary}\\\bbhline  T1-002T2-A104&122-988&Sue&Smith&\$300\\\hline T1-003T2-A101&198-877&Alice&Jones&\$100\\\bhline
\end{tabular}

\end{center}

\end{example}

\begin{remark}

We have seen that {\bf the right pushforward functor $\Pi_F$}, which arises naturally for any translation $F$ between schemas, {\bf automatically produces joins}.

\end{remark}

\begin{example}[Left Pushforward]\label{ex:left pushforward}

Let $F\taking\mcC\to\mcD$ be the translation given in Diagram (\ref{dia:translation}). In this example, we explore the data migration functor $\Sigma_F\taking\mcC\inst\to\mcD\inst$\footnote{We have not defined $\Sigma_F$ yet; this will be done in Section \ref{sec:left pushforward}.} by applying it to the $\mcC$-instance $I$ shown in Displays (\ref{dia:facts}) and (\ref{dia:leaves}). Note that our migration functor $\Sigma_F$ points in the same direction as $F$: it takes $\mcC$-instances to $\mcD$-instances. We now describe the $\mcD$-instance $\Sigma_F(I)$, which has four leaf tables $\Sigma_F(I)({\tt SSN})$, etc., and one fact table $\Sigma_F(I)({\tt T})$.

Instead of being a join, as in the case of $\Pi_F(I)$ above, the fact table ${\tt T}$ in instance $\Sigma_F(I)$ will be the union of {\tt T1} and {\tt T2}. One may wonder then how the category theoretic construction will deal with the fact that records in {\tt T1} do not have salary information and the records in {\tt T2} do not have SSN information. The answer is that the respective cells are {\em Skolemized}. In other words, the universal answer is to simply add a brand new ``variable" wherever one is needed in and downstream of {\tt T}. Thus in instance $\Sigma_F(I)$, table {\tt T} looks like this:

\begin{center}
\begin{tabular}{| l || l | l | l | l |}\bhline\multicolumn{5}{| c |}{{\tt T}}\\\bhline {\bf ID}&{\bf SSN}&{\bf First}&{\bf Last}&{\bf Salary}\\\bbhline  T1-001&115-234&Bob&Smith&T1-001.Salary\\\hline T1-002&122-988&Sue&Smith&T1-002.Salary\\\hline T1-003&198-877&Alice&Jones&T1-003.Salary\\\hline T2-A101&T2-A101.SSN&Alice&Jones&\$100\\\hline T2-A102&T2-A102.SSN&Sam&Miller&\$150\\\hline T2-A104&T2-A104.SSN&Sue&Smith&\$300\\\hline T2-A110&T2-A110.SSN&Carl&Pratt&\$200 \\\bhline
\end{tabular}
\end{center}

The Skolem variables (such as T1-001.Salary) can be equated with actual values later. They can also be equated with each other; for example we may know that T1-001.Salary=T2-002.Salary, without knowing the value of these salaries.

\end{example}

\begin{remark}

We have seen that {\bf the left pushforward functor $\Sigma_F$}, which arises naturally for any translation $F$ between schemas, {\bf automatically produces unions and automatically Skolemizes unknown values}.

\end{remark}

\subsection{Updatable views and linked multi-views}

Suppose we have a translation $F\taking\mcC\to\mcD$. In this case we can consider $\mcD$ as a view on $\mcC$ and consider $\mcC$ a view on $\mcD$. Unlike the classical version of views, our definition allows for arbitrarily many foreign keys between view tables; indeed, both $\mcC$ and $\mcD$ can be arbitrary schemas. Typical relational databases management systems such as SQL do not support ``linked multi-views", i.e. multiple view tables with foreign keys between them. For our data migration functors $\Pi_F,\Sigma_F$ and $\Delta_F$, this is no problem.

In fact, by the very nature of these three migration functors (i.e. by definition of the fact that they are functors), we have access to powerful theorems relating updates of $\mcC$-instances to updates of $\mcD$-instances. For example, given an instance $I$ on $\mcD$ whose $\Delta_F$-view is the instance $J=\Delta_F(I)$ on $\mcC$, and given an update $J\to J'$ on $\mcC$, there is a unique update $I\to\Pi_FJ'$, of instances on $\mcD$. A similar result holds for $\Sigma_F$ in place of $\Pi_F$: these facts follow from the fact that $\Pi$ (respectively $\Sigma_F$) is ``adjoint" to $\Delta_F$. See Section \ref{sec:data migration}.

The view update problem is often phrased as asking that ``the round trips are equivalences," (\cite{BCFGP}), which for us amounts to the composites $\Delta_F\Pi_F$ and $\Pi_F\Delta_F$ (respectively $\Sigma_F\Delta_F$ and $\Delta_F\Sigma_F$) being isomorphisms. This will only happe in case $F$ is an equivalence of categories. However, our data migration adjunctions provide view updates in more general circumstances, and these have provable formal properties (e.g. $\Sigma_F$ and $\Delta_F$ commute with inserts and $\Pi_F$ and $\Delta_F$ commute with deletes). But of course the best formal properties occur when $F$ is an equivalence.

The following example shows two things. First, it gives an example of a linked multi-view (foreign keys between views). Second, the translation $F$ is an equivalence of categories (a fact which relies essentially on the fact that $\mcC$ has path equivalences declared), and so the data migration functors $\Delta, \Pi$, and $\Sigma$ are also equivalences---they exhibit no information loss.

\begin{example}\label{ex:isomorphic objects}

Consider the two schemas drawn here: 

~\vspace{.2in}

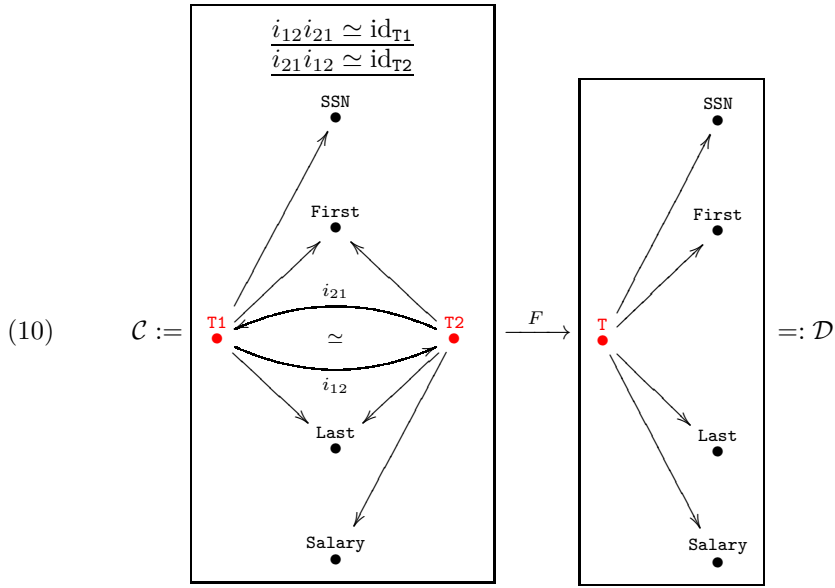
\begin{align}\label{dia:equivalence translation}\mcC:=\parbox{1.6in}{\vspace{-.4in}\fbox{\parbox{1.5in}{\begin{center}\ul{$i_{12}i_{21}\simeq\id_{{\tt T1}}$}\\\ul{$i_{21} i_{12}\simeq\id_{{\tt T2}}$}\end{center}\xymatrix{&\LTO{SSN}\\&\LTO{First}\\\color{red}{\LTO{T1}}\ar@/_1pc/[rr]_{i_{12}}\ar[uur]\ar[ur]\ar[dr]\ar@{}[rr]|{\simeq}&&\color{red}{\LTO{T2}}\ar@/_1pc/[ll]_{i_{21}}\ar[ul]\ar[dl]\ar[ddl]\\&\LTO{Last}\\&\LTO{Salary}}}}}\Too{F}\parbox{1in}{\fbox{\xymatrix{&\LTO{SSN}\\&\LTO{First}\\\color{red}{\LTO{T}}\ar[uur]\ar[ur]\ar[dr]\ar[ddr]\\&\LTO{Last}\\&\LTO{Salary}}}}=:\mcD\end{align}

The arrows $i_{12}$ and $i_{21}$, which are declared to be mutually inverse, ensure that the data which can be captured by schema $\mcC$ is equivalent to that which can be captured by schema $\mcD$. The translation $F$ sends $i_{12}$ and $i_{21}$ to the trivial path $\id_{\tt T}$ on {\tt T} (see Definitions \ref{def:graph} and \ref{def:translation}, and Example \ref{ex:follow up on equivalence translation}).

If table {\tt T} is as in Example \ref{ex:pullback} (\ref{dia:tableT}), then its pullback under $\Delta_F$ to an instance on $\mcC$ is similar to (\ref{dia:pullback facts}), but with additional columns $i_{12}$ and $i_{21}$ (because our schema has additional arrows $i_{12}$ and $i_{21}$):
\begin{align}\label{dia:equivalence pullback}
\text{\scriptsize\begin{tabular}{| l || l | l | l | l |}\bhline\multicolumn{5}{| c |}{{\tt T1}}\\\bhline {\bf ID}&{\bf SSN}&{\bf First}&{\bf Last}&{\bf $i_{12}$}\\\bbhline XF667T1&115-234&Bob&Smith&A67\\\hline XF891T1&122-988&Sue&Smith&A91\\\hline XF221T1&198-877&Alice&Jones&A21\\\bhline
\end{tabular}}\hsp
\text{\scriptsize\begin{tabular}{| l || l | l | l | l |}\bhline\multicolumn{5}{| c |}{{\tt T2}}\\\bhline {\bf ID}&{\bf First}&{\bf Last}&{\bf Salary}&{\bf $i_{21}$}\\\bbhline A21&Alice&Jones&\$100&XF221T1\\\hline A67&Bob&Smith&\$250&XF667T1\\\hline A91&Sue&Smith&\$300&XF891T1\\\bhline
\end{tabular}}
\end{align}
The foreign key columns $i_{12}$ and $i_{21}$ on the $\mcC$-view keep track of the data necessary for successful round-tripping. An update to a $\mcD$-instance will yield a corresponding update to a $\mcC$-instance and vice versa. The fact that $F$ is an {\em equivalence of categories} implies that $\Sigma_F,\Pi_F,$ and $\Delta_F$ are also equivalences of categories, and roundtrip isomorphisms will hold for all possible updates.

\end{example}

\subsection{Interoperability with RDF data}\label{sec:RDF}

The {\em Resource Descriptive Framework (RDF)} is the semantic web standard data format \cite{KC}. The basic idea is to encode all facts in terms of basic \begin{center}(Subject Predicate Object)\end{center} triples, such as (Bob hasMother Sue). There are papers devoted to understanding the transformation from relational databases to RDF triple stores, and vice versa (\cite{ASX},\cite{KT}). In this section we will assume a basic familiarity with the jargon of that field, such as {\em URI} (uniform resource identifier).

Category-theoretically, the formulation of RDF triple stores is quite simple. Given a schema $\mcC$, a triple store over $\mcC$ is a category $\mcS$ (representing the triples) and a functor $\pi\taking\mcS\to\mcC$ (representing their types). The objects in $\mcS$ are URIs; the arrows in $\mcS$ are triples $$\LTO{Subject}\To{\tn{Predicate}}\LTO{Object}$$ Given an object $c\in\mcC$ in the schema, the inverse-image $\pi^\m1(c)\ss\mcS$ consists of all URIs of type $c$. Given an arrow $f\taking c\to c'$ in $\mcC$, the inverse image is the $f$-relation between $\pi^\m1(c)$ and $\pi^\m1(d)$.\footnote{This relation can be functional or inverse functional, as dictated by the RDF schema; the subject can be understood category-theoretically by so-called ``lifting constraints" (see \cite{Sp4}).}

There is a basic category-theoretic operation that converts a relational database instance into an RDF triple store (and a straightforward inverse as well, converting an RDF triple store into a relational database instance). It is called the {\em Grothendieck construction}. Consider for example the instance $I$ from Example \ref{ex:famous schema} Display (\ref{dia:employee}): 
\begin{align*}&\begin{tabular}{| l || l | l | l | l |}\bhline\multicolumn{5}{| c |}{{\tt Employee}}\\\bhline {\bf ID}&{\bf First}&{\bf Last}&{\bf Mgr}&{\bf isIn}\\\bbhline 101&David&Hilbert&103&q10\\\hline 102&Bertrand&Russell&102&x02\\\hline 103&Alan&Turing&103&q10\\\bhline\end{tabular}&\begin{tabular}{| l || l | l |}\bhline\multicolumn{3}{| c |}{{\tt Department}}\\\bhline {\bf ID}&{\bf Name}&{\bf Secr}\\\bbhline q10&Sales&101\\\hline x02&Production&102\\\bhline\end{tabular}\end{align*}
Taking the Grothendieck construction yields the following triple store $\mcS=Gr(I)$, where each arrow designates a RDF triple, as above:
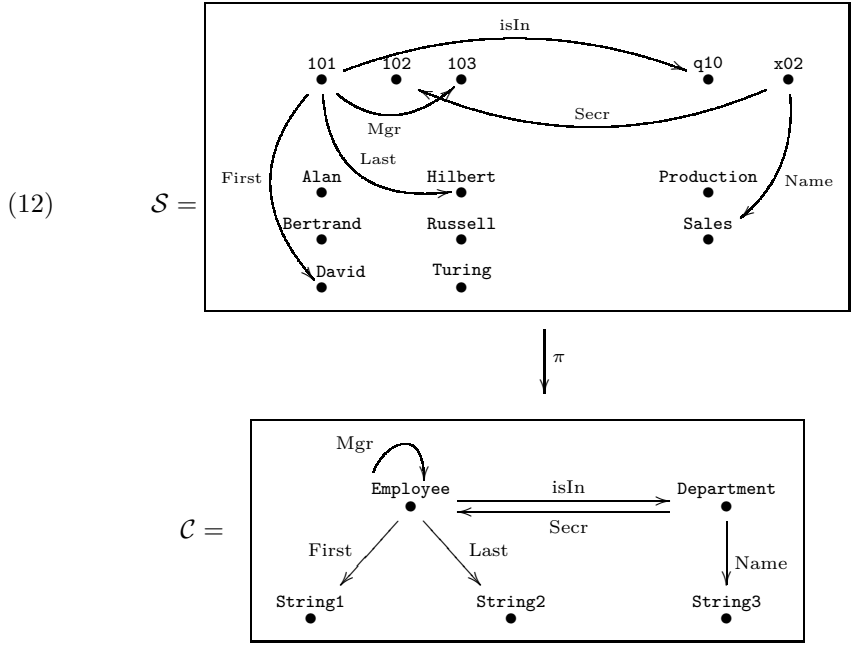
\begin{align}\label{dia:Grothendieck}
&~\hspace{-.2in}\mcS=\parbox{2.6in}{\fbox{\xymatrix@=.1pt{\LTO{101}\ar@/_1.5pc/[ddddddd]+<-3pt,3pt>_{\tin{First}}\ar@/_1.5pc/[dddddrr]+<-4pt,0pt>^{\tin{Last}}\ar@/_1pc/[rr]+<-3pt,-3pt>_{\tin{Mgr}}\ar@/^1.3pc/[rrrrr]^{\tin{isIn}}&\LTO{102}&\LTO{103}&&&\LTO{q10}&\LTO{x02}\ar@/^1.5pc/[lllll]_{\tin{Secr}}\ar@/^1pc/[ldddddd]^{\tin{Name}}\\\\\\\\\\\LTO{Alan}&&\LTO{Hilbert}&\hsp&\hsp&\LTO{Production}\\\LTO{Bertrand}&&\LTO{Russell}&&&\LTO{Sales}\\\LTO{\hspace{.2in}David}&&\LTO{Turing}}}\\\xymatrix{\hspace{1.3in}&\ar[d]^\pi\\&~}}\\
\nonumber&\mcC=\parbox{2.9in}{\hspace{.1in}\fbox{
			\xymatrix@=4pt{&\LTO{Employee}\ar@<.5ex>[rrrrr]^{\tn{isIn}}\ar@(l,u)[]+<5pt,10pt>^{\tn{Mgr}}\ar[dddl]_{\tn{First}}\ar[dddr]^{\tn{Last}}&&&&&\LTO{Department}\ar@<.5ex>[lllll]^{\tn{Secr}}\ar[ddd]^{\tn{Name}}\\\\\\\LTO{String1}&&\LTO{String2}&~&~&~&\LTO{String3}
			}
		}}\vspace{1in}
\end{align}
In Display (\ref{dia:Grothendieck}), ten arrows have been left out of the picture of $S$, (e.g. the arrow $\LTO{102}\Too{\tn{Last}}\LTO{Russell}$ is not pictured) for readability reasons. The point is that the RDF triple store associated to instance $I$ is nicely represented using the standard Grothendieck construction.

\subsection{Close connection between data and program}\label{sec:data and program}

Currently, there is an ``impedance mismatch" between databases and programming languages; their respective formulations and underlying models do not cohere as well as they should (\cite{CI}). Whereas the programming languages (PL) community has embraced category theory for the conceptual clarity and expressive power it brings, most database theorists tend to concentrate on practical considerations, such as speed, reliability, and scalability. The importance of databases in the modern world cannot be overstated, and yet in order for databases to reach their full potential, better theoretical integration with applications must be developed.

As stated in the Introduction (Section \ref{sec:intro}), the first goal of this paper is to present a straightforward model of databases under which every theorem about small categories becomes a theorem about databases. Thus the favorite category of PL theorists, namely the category $\Type$ of types and terms (for some fixed $\lambda$-calculus, see \cite[Section 6.5]{Awo}), is a kind of infinite database schema: its tables correspond to types and its foreign key columns correspond to terms. Of course, unlike real-world databases in which tables model real-world entities and their relationships (such as people and their heights), the schema $\Type$ models mathematical entities and their relationships (such as integers and their factorials). However, these ideas clearly live in the same platonic realm, so to speak, and this notion is expressed by saying that both database schemas and $\Type$ can be considered as categories and related by functors.

This leads to nice integration between data and program. For example many spreadsheet capabilities, such summing up the values in two columns to get the value in a third, can be included at the schema level. At this point in the paper, we do not yet have the necessary machinery to show exactly how that should work (see Section \ref{sec:DT as NT}), but in the following diagram one can see the schematic presentation of the relevant subcategory of $\Type$: \begin{align}\label{dia:plus}P:=\fbox
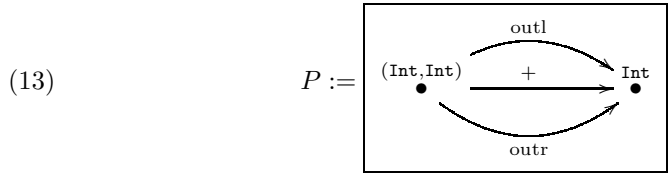
{\xymatrix{\LTO{(Int, Int)}\ar[rr]^{+}\ar@/^1.5pc/[rr]^{\tn{outl}}\ar@/_1.5pc/[rr]_{\tn{outr}}&&\LTO{Int}}}\end{align} The point is to simultaneously see two different things within this one diagram (like an optical illusion). The first thing to see in Diagram (\ref{dia:plus}) is a database schema. In schema $P$, we have two tables: 
\begin{align*}&\footnotesize\begin{tabular}{| l || l | l | l |}\bhline\multicolumn{4}{| c |}{{\tt (Int,Int)}}\\\bhline {\bf ID}&{\bf outl}&{\bf outr}&{\bf +}\\\bbhline P0c0&0&0&0\\\hline P1c0&1&0&1\\\hline P1c1&1&1&2\\\hline P0c1&0&1&1\\\hline P2c0&2&0&2\\\hline P2c1&2&1&3\\\hline P2c2&2&2&4\\\hline P1c2&1&2&3\\\hline P0c2&0&2&2\\\hline\vdots&\vdots&\vdots&\vdots\\\bhline\end{tabular}&
\begin{tabular}{| l ||}\bhline\multicolumn{1}{| c |}{{\tt Int}}\\\bhline {\bf ID}\\\bbhline 0\\\hline 1\\\hline 2\\\hline 3\\\hline 4\\\hline \vdots\\\bhline\end{tabular}\end{align*}\normalsize
The second thing to see in Diagram (\ref{dia:plus}) is a subcategory $P\ss\Type$, i.e. a close connection to standard PL theory. The same category $P$ is viewed extensionally in the context of databases and intentionally in the context of programs. This dual citizenship of categories makes category theory a good candidate for solving the impedance mismatch between databases and programming languages.

\section{Definitions}\label{sec:definitions}

In this section, our main goal is to define a category of schemas and translations, and to show that it is equivalent to $\Cat$, the category of small categories. Along the way we will define the category of instances on a given schema. Finally, we will give a dictionary that one can use to translate between database concepts (e.g. found in \cite{EN}) and category-theoretic concepts (e.g. found in \cite{Mac}).

\subsection{Some references}\label{sec:references}

Throughout this section, we will assume the reader has familiarity with the fundamental notions of category theory: objects, morphisms, and commutative diagrams within a category; as well as categories, functors, and natural transformations. There are many good references on category theory, including \cite{LS}, \cite{Sic}, \cite{Pie}, \cite{BW1}, \cite{Awo}, and \cite{Mac}; the first and second are suited for general audiences, the third and fourth are suited for computer scientists, and the fifth and sixth are suited for mathematicians (in each class the first reference is easier than the second). One may also see \cite{SK} for a different perspective.

\subsection{Graphs, Paths, Schemas, and Instances}

A graph (sometimes called a directed multi-graph) is a collection of vertices and arrows, looking something like this: \begin{align}\label{dia:graph}\parbox{1.5in}{\fbox
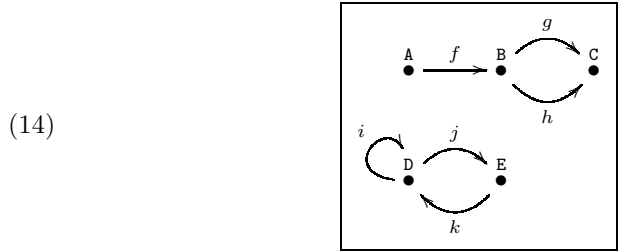
{\xymatrix{\LTO{A}\ar[r]^f&\LTO{B}\ar@/_1pc/[r]_h\ar@/^1pc/[r]^g&\LTO{C}\\\LTO{D}\ar@(l,u)[]^i\ar@/^1pc/[r]^j&\LTO{E}\ar@/^1pc/[l]^k}}}\end{align} This is one graph with two connected components; it has five vertices and six arrows. 

\begin{definition}\label{def:graph}

A {\em graph} $G$ is a sequence $G=(A,V,src,tgt)$, where $A$ and $V$ are sets (respectively called {\em the set of arrows} and {\em the set of vertices} of $G$), and $src\taking A\to V$ and $tgt\taking A\to V$ are functions (respectively called {\em the source function} and {\em the target function} for $G$). If $a\in A$ is an arrow with source $src(a)=v$ and target $tgt(a)=w$, we draw it as $$v\Too{a}w.$$

\end{definition}

\begin{definition}

Let $G=(A,V,src,tgt)$ be a graph. A {\em path of length $n$} in $G$, denoted $p\in\Path_G^{(n)}$ is a head-to-tail sequence \begin{align}\label{dia:path}p=(v_0\To{a_1}v_1\To{a_2}v_2\To{a_3}\ldots\To{a_n}v_n)\end{align} of arrows in $G$. In particular, $\Path_G^{(1)}=A$ and $\Path_G^{(0)}=V$; we refer to the path of length 0 on vertex $v$ as the {\em trivial path on $v$} and denote it by $\id_v$. We denote by $\Path_G$ the set of all paths on $G$, $$\Path_G:=\bigcup_{n\in\NN}\Path_G^{(n)}.$$ Every path $p\in\Path_G$ has a source vertex and a target vertex, and we may abuse notation and write $src,tgt\taking\Path_G\to V$. If $p$ is a path with $src(p)=v$ and $tgt(p)=w$, we may denote it by $p\taking v\to w$. Given two vertices $v,w\in V$, we write $\Path_G(v,w)$ to denote the set of all paths $p\taking v\to w$.

There is a composition operation on paths. Given a path $p\taking v\to w$ and $q\taking w\to x$, we define the composition, denoted $pq\taking v\to x$ in the obvious way. In particular, if $p$ (resp. $r$) is the trivial path on vertex $v$ (resp. vertex $w$) then for any path $q\taking v\to w$, we have $pq=q$ (resp. $qr=q$). Thus, for clarity, we may always denote a path as beginning with a trivial path on its source vertex; e.g. the path $p$ from Diagram (\ref{dia:path}) may be denoted $p=\id_{v_0}a_1a_2\cdots a_n$.

\end{definition}

\begin{example}

In Diagram (\ref{dia:graph}), there are no paths from $A$ to $D$, one path ($f$) from $A$ to $B$, two paths ($fg$ and $fh$) from $A$ to $C$, and infinitely many paths $\{i^{p_1}(jk)^{q_1}\cdots i^{p_n}(jk)^{q_n}\;|\;n,p_1,q_1,\ldots,p_n,q_n\in\NN\}$) from $D$ to $D$.

\end{example}

We now define the notion of categorical equivalence relation on the set of paths of a graph. Such an equivalence relation (in addition to being reflexive, symmetric, and transitive) has two sorts of additional properties: equivalent paths have the same source and target, and the composition of equivalent paths with other equivalent paths must yield equivalent paths. Formally we have Definition \ref{def:CPER}.

\begin{definition}\label{def:CPER}

Let $G=(A,V,src,tgt)$ be a graph. A {\em categorical path equivalence relation} (or {\em CPER}) on $G$ is an equivalence relation $\simeq$ on $\Path_G$ that has the following properties: 
\begin{enumerate}
\item If $p\simeq q$ then $src(p)=src(q)$.
\item If $p\simeq q$ then $tgt(p)=tgt(q)$.
\item Suppose $p,q\taking b\to c$ are paths, and $m\taking a\to b$ is an arrow. If $p\simeq q$ then $mp\simeq mq$. 
\item Suppose $p,q\taking a\to b$ are paths, and $n\taking b\to c$ is an arrow. If $p\simeq q$ then $pn\simeq qn$.
\end{enumerate}
\end{definition}

\begin{lemma}

Suppose that $G$ is a graph and $\simeq$ is a CPER on $G$. Suppose $p\simeq q\taking a\to b$ and $r\simeq s\taking b\to c$. Then $pr\simeq qs$.

\end{lemma}

\begin{proof}

The picture to have in mind is this: $$
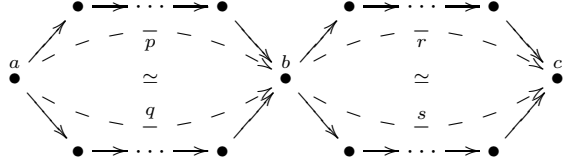
\xymatrix@=13pt{&\bullet\ar[r]&\cdots\ar[r]&\bullet\ar[dr]&&\bullet\ar[r]&\cdots\ar[r]&\bullet\ar[dr]\\\LMO{a}\ar@{}[rrrr]|{\simeq}\ar[ur]\ar[dr]\ar@{-->}@/^1.5pc/[rrrr]_p\ar@{-->}@/_1.5pc/[rrrr]^q&&&&\LMO{b}\ar@{}[rrrr]|{\simeq}\ar[ur]\ar[dr]\ar@{-->}@/^1.5pc/[rrrr]_r\ar@{-->}@/_1.5pc/[rrrr]^s&&&&\LMO{c}\\&\bullet\ar[r]&\cdots\ar[r]&\bullet\ar[ur]&&\bullet\ar[r]&\cdots\ar[r]&\bullet\ar[ur]}$$ Applying condition (3) from Definition \ref{def:CPER} to each arrow in path $p$, it follows by induction that $pr\simeq ps$. Applying condition (4) to each arrow in path $s$, it follows similarly that $ps\simeq qs$. Because $\simeq$ is an equivalence relation, it follows that $pr\simeq qs$. 

\end{proof}

\begin{definition}\label{def:categorical schema}

A {\em categorical schema} $\mcC$ consists of a pair $\mcC:=(G,\simeq)$ where $G$ is a graph and $\simeq$ is a categorical path equivalence relation on $G$. We sometimes refer to a categorical schema as simply a {\em schema}.

\end{definition}

\begin{example}\label{ex:self email}

Consider the schema, i.e. the graph together with the indicated equivalence\footnote{More precisely, consider the graph with the categorical equivalence relation generated by the set $\{fg=fh\}$.}, pictured in the box below: $$\mcC:=\fbox{\parbox{1.2in}{\begin{center}\ul{$fg\simeq fh$}\end{center}~\\
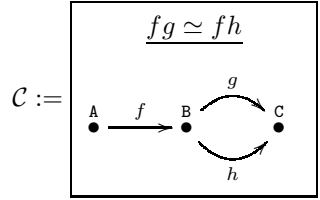
\xymatrix{\LTO{A}\ar[r]^f&\LTO{B}\ar@/^1pc/[r]^g\ar@/_1pc/[r]_h&\LTO{C}}}}$$ This schema models, for example, the phenomenon of sending an email to oneself. Suppose we populate $B$ with emails, $C$ with people, $g$ and $h$ with the sender and receiver fields, respectively. Then for $fg$ to equal $fh$ we must have that senders equal receivers on the image of $f$, and thus the subset of self-emails is a perfect fit for {\tt A}. See example \ref{ex:instance of self-email}.

\end{example}

More on the subject of categorical schemas, including a picture of a schema and an associated set of tables, can be found in Example \ref{ex:famous schema}. 

In the following, we will define what it means to be an instance of a categorical schema $\mcC$. We consider the case in which our instances are set-models of $\mcC$, but the same idea works in much more generality (see Definition \ref{def:category of instances}).

\begin{definition}\label{def:instance}

Let $\mcC:=(G,\simeq)$ be a categorical schema, where $G=(A,V,src,tgt)$. An {\em instance on $\mcC$}, denoted $I$, consists of the following 
\begin{enumerate}
\item For every vertex $v\in V$, a set $I(v)$.
\item For every arrow $a\taking v\to v'$ in $A$, a function $I(a)\taking I(v)\to I(v')$
\item For every path equivalence $p\simeq q$ a guarantee that the equation $I(p) = I(q)$ holds.\footnote{Once $I$ is defined on arrows, as it is in item (2), we can extend it to paths in the obvious way: if $p=a_1a_2\cdots a_n$, then the function $I(p)$ is the composition $I(p)=I(a_1)I(a_2)\cdots I(a_n)$.}

\end{enumerate}

\end{definition}

\begin{example}\label{ex:instance of self-email}

We now return to Example \ref{ex:self email}, and write down a sample instance $I$ for schema $\mcC=(G,\simeq)$.
\begin{align*}
\begin{tabular}{| l || l |}\bhline\multicolumn{2}{| c |}{{\tt A}}\\\bhline {\bf ID}&{\bf $f$}\\\bbhline SEm1207&Em1207\\\hline SEm1210&Em1210\\\hline SEm1211&Em1211\\\bhline\end{tabular}&\hsp
\begin{tabular}{| l || l | l |}\bhline\multicolumn{3}{| c |}{{\tt B}}\\\bhline {\bf ID}&{\bf $g$}&{\bf $h$}\\\bbhline Em1206&Bob&Sue\\\hline Em1207 &Carl&Carl\\\hline Em1208&Sue & Martha\\\hline Em1209&Chris&Bob\\\hline Em1210&Chris&Chris\\\hline Em1211&Julia&Julia\\\hline Em1212&Martha&Chris\\\bhline\end{tabular}\hsp
\begin{tabular}{| l ||}\bhline\multicolumn{1}{| c |}{{\tt C}}\\\bhline {\bf ID}\\\bbhline Bob\\\hline Carl\\\hline Chris\\\hline Julia\\\hline Martha\\\hline Sue\\\bhline\end{tabular}
\end{align*}\normalsize 
For each vertex $v$ in $G$, the set $I(v)$ is given by the set of rows in the corresponding table (e.g. $I({\tt A})=\{\text{SEm1207, SEm1210, SEm1211}\}$). For each arrow $a\taking v\to w$ in $G$ the function $I(a)\taking I(v)\to I(w)$ is also evident as a column in the table. For example, $I(g)\taking I({\tt B})\to I({\tt C})$ sends Em1206 to Bob, etc. Finally, the path equivalence $f g=f h$ is borne out in the fact that for every row-id in table {\tt A}, following $f$ then $g$ returns the same result as following $f$ then $h$.
\end{example}

\subsection{Translations}

A translation is a mapping from one categorical schema to another. Vertices are sent to vertices, arrows are sent to paths, and all path equivalences are preserved. More precisely, we have the following definition.

\begin{definition}\label{def:translation}

Let $G=(A_G,V_G,src_G,tgt_G)$ and $H=(A_H,V_H,src_H,tgt_H)$ be graphs (see Definition \ref{def:graph}), and let $\mcC=(G,\simeq_\mcC)$ and $\mcD=(H,\simeq_\mcD)$ be categorical schemas. A {\em translation $F$ from $\mcC$ to $\mcD$}, denoted $F\taking\mcC\to\mcD$ consists of the following constituents:
\begin{enumerate}[\hspace{.6cm}(1)]
\item a function $V_F\taking V_G\to V_H$, and
\item a function $A_F\taking A_G\to\Path_H$
\end{enumerate}
subject to the following conditions:
\begin{enumerate}[\hspace{.6cm}(a)]
\item the function $A_F$ preserves sources and targets; in other words, the following diagrams of sets commute: $$\xymatrix{A_G\ar[r]^{A_F}\ar[d]_{src_G}&\Path_H\ar[d]^{src_H}\\V_G\ar[r]_{V_F}&V_H}\hspace{.4in} \xymatrix{A_G\ar[r]^{A_F}\ar[d]_{tgt_G}&\Path_H\ar[d]^{tgt_H}\\V_G\ar[r]_{V_F}&V_H}$$
\item the function $A_F$ preserves path equivalences.\footnote{This is easier to understand conceptually than to write down.} Precisely, suppose we are given lengths $m,n\in\NN$ and paths $p=id_{v_0} f_1 f_2 \cdots f_m$ and $q=id_{v_0} g_1 g_2 \cdots g_n$ in $G$. Let $v'_0=V_F(v_0)$ and for each $i\leq m$ (resp. $j\leq n$), let $f'_i=A_F(f_i)$ (resp. $g'_j=A_F(g_j)$), and let $p'=\id_{v'_0} f'_1 f'_2 \cdots f'_m$ (resp. $q'=\id_{v'_0} g'_1 g'_2 \cdots g'_n$). If $p\simeq_\mcC q$ then $p'\simeq_\mcD q'$. 
\end{enumerate}

Two translations $F,F'\taking\mcC\to\mcD$ are considered identical if they agree on vertices (i.e. $V_F=V_{F'}$) and if, for every arrow $f$ in $\mcC$, there is a path equivalence $$A_F(f)\simeq_\mcD A_{F'}(f).$$

\end{definition}

In the following two examples we will reconsider translations discussed in Section \ref{sec:examples}.

\begin{example}

Recall the mapping $F$ given in Diagram (\ref{dia:translation}). The schemas $\mcC$ and $\mcD$ are just graphs in the sense that there are no declared path equivalences in either of them. The mapping $F$ sends vertices in $\mcC$ to vertices in $\mcD$ and arrows in $\mcC$ to arrows in $\mcD$. Since an arrow is a particular sort of path, and since there are no path equivalences to be preserved, $F\taking\mcC\to\mcD$ is indeed a translation.

\end{example}

\begin{example}\label{ex:follow up on equivalence translation}

Recall the mapping $F$ given in Example \ref{ex:isomorphic objects}, Diagram (\ref{dia:equivalence translation}). In this setup, $\mcC$ has declared path equivalences and $\mcD$ does not; however $\mcD$ still has a categorical path equivalence relation $\simeq_\mcD$ on it, the minimal reflexive relation. The mapping $F$ on vertices ($V_F$) is self-explanatory; the only arrows on which $A_F$ is not self-explanatory are $i_{12}$ and $i_{21}$, both of which are sent to the trivial path $\id_{\tt T}$ on vertex {\tt T}. 

Because $V_F({\tt T1})=V_F({\tt T2})={\tt T}$, it is clear that $A_F$ preserves sources and targets. The path equivalence $i_{12} i_{21}=\id_{\tt T1}$ and $i_{21} i_{12}=\id_{\tt T2}$ are preserved because $A_F(i_{12})=A_F(i_{21})=\id_{\tt T}$, and the concatenation of a trivial path with any path $p$ yields $p$.
\end{example}

\subsection{The equivalence $\Sch\simeq\Cat$}

We assume familiarity with categories and functors, and in particular the category $\Cat$ of small categories and functors (a list of references is given in Section \ref{sec:references}). In this section we will define the category $\Sch$ and show it is equivalent to $\Cat$. It is this result that justifies our advertisement in the introduction that ``every theorem about small categories becomes a theorem about databases".

\begin{definition}

Recall the notions of categorical schemas and translations from Definitions \ref{def:categorical schema} and \ref{def:translation}. The {\em category of categorical schemas}, denoted $\Sch$, is the category whose objects are categorical schemas and whose morphisms are translations.

\end{definition}

{}

\begin{construction}[From schema to category]

We will define a functor $L\taking\Sch\to\Cat$. Let $\mcC=(G,\simeq_\mcC)$ be a categorical schema, where $G=(A,V,src,tgt)$. Define $\mcC'$ to be the free category with objects $V$ generated by arrows $A$. Define $L(\mcC)\in\Cat$ to be the category defined as the quotient of $\mcC'$ by the equivalence relation $\simeq_\mcC$ (see \cite[Section 2.8]{Mac}). This defines $L$ on objects of $\Sch$.

Given a translation $F\taking\mcC\to\mcD$, there is an induced functor on free categories $F'\taking\mcC'\to\mcD'$, sending each generator $f\in A$ to the morphism in $\mcD'$ defined as the composite of the path $A_F(f)$. The preservation of path equivalence ensures that $F'$ descends to a functor $L(F)\taking L(\mcC)\to L(\mcD)$ on quotient categories. This defines $L$ on morphisms in $\Sch$. It is clear that $L$ preserves composition, so it is a functor.

\end{construction}

\begin{construction}[From category to schema]

We will define a functor $R\taking\Cat\to\Sch$. Let $\mcC$ be a small category with object set $\Ob(\mcC)$, morphism set $\Mor(\mcC)$, source and target functions $s,t\taking\Mor(\mcC)\to\Ob(\mcC)$, and composition law $\circ\taking\Mor(\mcC)\cross_{\Ob(\mcC)}\Mor(\mcC)\to\Mor(\mcC)$. Let $R(\mcC)=(G,\simeq)$ where $G$ is the graph $$G=(\Mor(\mcC),\Ob(\mcC),s,t)$$ and with $\simeq$ defined as follows: for all $f,g\in\Mor(\mcC)$ with $t(f)=s(g)$ we put \begin{align}\label{dia:composition law}f g\simeq (g\circ f).\end{align} This defines $R$ on objects of $\Cat$.

A functor $F\taking\mcC\to\mcD$ induces a translation $R(F)\taking R(\mcC)\to R(\mcD)$, because vertices are sent to vertices, arrows are sent to arrows, and path equivalence is preserved by (\ref{dia:composition law}) and the fact that $F$ preserves the composition law. This defines $R$ on morphisms in $\Cat$. It is clear that $R$ preserves compositions, so it is a functor.

\end{construction}

\begin{theorem}\label{thm:equivalence of categories}

The functors $$\xymatrix{L\taking\Sch\ar@<.5ex>[r]&\Cat\!:R\ar@<.5ex>[l]}$$ are mutually inverse equivalences of categories.

\end{theorem}

\begin{proof}[Sketch of proof]

It is clear that there is a natural isomorphism $\epsilon\taking\id_\Cat\to L\circ R$; i.e. for any category $\mcC$, there is an isomorphism $\mcC\iso L(R(\mcC))$. Thus the functor $L$ is essentially surjective. We first show that $L$ is fully faithful. 

Choose schemas $X$ and $Y$, and suppose $X=(A_X,V_X,src_X,tgt_X)$; we must show that the function $L_1\taking\Hom_{\Sch}(X,Y)\to\Hom_{\Cat}(LX,LY)$ is a bijection. It is clearly injective. To show that it is surjective, choose a functor $G\taking LX\to LY$; we will define a translation $F\taking X\to Y$ with $L_1(F)=G$. Define $F$ on vertices of $X$ as $G$ is defined on objects of $LX$. Define $F$ on arrows of $X$ via the function $A_X\to\Path_X\to\Mor(LX)\To{G}\Mor(LY)$, and choose a representative for its equivalence class from $\Path_{LY}$ (note that any two choices result in the same translation: see Definition \ref{def:translation}). Two equivalent paths in $X$ compose to the same element of $\Mor(LX)$, so $F$ preserves path equivalence. This defines $F$, completing the proof that $L$ an equivalence of categories.

By similar reasoning one proves that $R$ is fully faithful, and concludes that it is inverse to $L$.

\end{proof}

\subsection{The category of instances on a schema}

Given Theorem \ref{thm:equivalence of categories}, the compound notion of categorical schemas and translations is equivalent to that of categories and functors. In the remainder of the paper, we elide the difference between $\Sch$ and $\Cat$, using nomenclature from each interchangeably. 

One sees easily (c.f. Definition \ref{def:instance}) that an instance $I$ on a schema $\mcC$ is the same thing as a functor $\mcC\to\Set$, where $\Set$ is the category of sets. Thus we have a ready-made concept of morphisms between instances: natural transformations of functors. This is an established notion in database literature, often called a {\em homomorphism} of instances (see e.g. \cite{DNR}).

\begin{definition}\label{def:category of instances}

Let $\mcC$ be a schema and let $I,J\taking\mcC\to\Set$ be instances on $\mcC$. A {\em morphism $m$ from $I$ to $J$}, denoted $m\taking I\to J$, is simply a natural transformation between these functors. We define {\em the category of instances on $\mcC$}, denoted $\mcC\inst$, to be the category of instances and morphisms, as above.

More generally, let $\SS$ denote any category; we define {\em the category of $\SS$-valued instances on $\mcC$}, denoted $\mcC\inst_\SS:=\SS^\mcC$, to be the category whose objects are functors $\mcC\to\SS$ and whose morphisms are natural transformations. We refer to $\SS$ as the {\em value category} in this setup.

\end{definition}

\begin{remark}

It appears that programming language theorists do not include homomorphisms between instances in their conception of database instances as elements of a type, preferring instead to work with just the {\em set} $\Ob(\mcC\inst)$ of instances on a schema $\mcC$. Doing so makes it easier to define aggregate functions, such as sums and counts; see e.g. \cite{LT}. 

\end{remark}

In the rest of the paper, we will generally work with $\mcC\inst$, the category of $\Set$-valued instances. However, most of the results go through more generally for $\mcC\inst_\SS$, provided that $\SS$ is complete and cocomplete (i.e. has all small limits and all small colimits). Obviously, given a functor $\SS\to\SS'$ there is an induced functor $\mcC\inst_\SS\to\mcC\inst_{\SS'}$, so the choice of value-category can be changed without much cost.

\begin{example}

Given a schema $\mcC$, there are many categories $\SS$, other than $\SS=\Set$, for which one might be interested in $\mcC\inst_\SS$. For example, given a lambda calculus, the associated category $\SS=\Type$ of types and terms is a good choice (\cite[Section 6.5]{Awo}). One can also use the category $\Fin$ of finite sets, $\cpo$ of complete partial orders, $\Cat$ of small categories, or $\Top$ of topological spaces. 

The choice of value-category is based on how one chooses to view the collection of rows in each table. We usually consider this collection to be a set, but for example one can imagine instead a {\em topological space} of rows, and in this case each column would consist of a continuous map from one space to another.

\end{example}

Toposes, invented by Grothendieck and Verdier \cite{GV} and extended by Lawvere \cite{Law}, are categories that mimic the category of sets in several important ways (see \cite{MM}). In Proposition \ref{prop:functor categories}, we show that the instance categories are often toposes.

\begin{proposition}\label{prop:functor categories}

If $\SS=\Set$ then for any schema $\mcC\in\Sch$, the category $\mcC\inst=\mcC\inst_\Set$ is a topos. If $\mcC$ is a finite category then for any topos $\SS$ (e.g. $\SS=\Fin$, the category of finite sets), the category $\mcC\inst_\SS$ is a topos. If $\SS$ is complete (resp. cocomplete) then $\mcC\inst_\SS$ is also complete (resp. cocomplete).

\end{proposition}

\begin{proof}

The first pair of claims are \cite[A.2.1.3]{JoP}. The second pair of claims are found in \cite[Theorem 2.15.2]{Bor1} and \cite[Proposition 8.8]{Awo}, respectively.

\end{proof}

Given a database instance $I$, updates on $I$ include deletion of rows, insertion of rows, splitting (one row becoming two), and merging (two rows becoming one). In fact, we classify insertions and merges together as {\em progressive updates} and we classify deletions and splits together as {\em regressive updates}. Then every update can be considered as a regressive update followed by a progressive update.

\begin{definition}

Let $\mcC$ be a schema and $I\in\mcC\inst$ an instance. A {\em progressive update on $I$} consists of an instance $J$ and a natural transformation $p\taking I\to J$. A {\em regressive update on $I$} consists of an instance $J$ and a natural transformation $r\taking J\to I$. That is, a regressive update is just a progressive update in reverse. An {\em update} is a finite sequence of progressive and regressive updates.

\end{definition}

\begin{proposition}

Let $\mcC$ be a schema and $I$ an instance on $\mcC$. Any update on $I$ can be obtained as a single regressive update followed by a single progressive update.

\end{proposition}

\begin{proof}

The composition of two progressive (resp. regressive) updates is a progressive (resp. regressive) update. Hence any update on $I_0:=I$ can be written as a diagram $D$ in $\mcC\inst$:
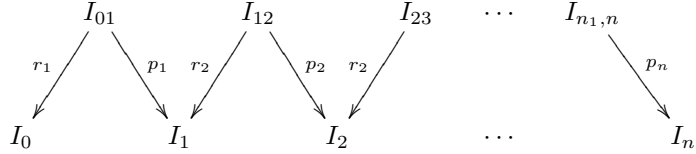
$$\xymatrix@=13pt{&I_{01}\ar[ddl]_{r_1}\ar[ddr]^{p_1}&&I_{12}\ar[ldd]_{r_2}\ar[rdd]^{p_2}&&I_{23}\ar[ddl]_{r_2}&\cdots&I_{n_1,n}\ar[ddr]^{p_n}
\\\\I_0&&I_1&&I_2&&\cdots&&I_n}$$ But the limit of this diagram (which can be taken, if one wishes, by taking fiber products such as $I_{j-1,j}\to I_j\from I_{j,j+1}$, and repeating $\frac{n^2-n}{2}$ times), is what we need: $$I_0\from\lim D\to I_n.$$

\end{proof}

\begin{remark}

Another way to understand deletes is via filtering---one filters out all rows of a certain form. Filtering will be discussed in Section \ref{sec:filtering}.

\end{remark}

\subsection{Grothendieck construction}\label{sec:Grothendieck}

In Section \ref{sec:RDF} we showed how to convert an instance $I\taking\mcC\to\Set$ to a new category $Gr(I)$, called the {\em Grothendieck construction} or {\em category of elements} of $I$. This construction models the conversion from relational to RDF forms of data. There is a reverse construction that is described in \cite[Proposition 2.3.9]{Sp4}.

We note here that there is a more general Grothendieck construction that may be useful in the context of federated databases. In programming languages theory, one may hear of a category of {\em kinds}, each object of which is itself a category of types. Here, each kind is analogous to a schema, and each type in that kind is analogous to a table in that schema. Given a category of kinds,  we can ``throw all their types together" by applying the generalized Grothendieck construction. This is akin to taking a federated database (i.e. a schema of related schemas) and merging them all into a single grand schema. One can apply this construction at the data level as well, merging all the instances into an instance over the single grand schema.

The version of the Grothendieck construction given in Section \ref{sec:RDF} is for functors $I\taking\mcC\to\Set$. Each object $c\in\Ob(\mcC)$ corresponds to a table whose set of rows is $I(c)$. One can find a description of this construction in \cite[Section 1.5]{MM}. The version of the Grothendieck construction in which federated schemas are combined into one big schema is for functors $D\taking\mcC\to\Cat$. Each object $c\in\Ob(\mcC)$ corresponds to a database whose category of tables is $D(c)$. One can find a description of this construction in \cite[B.1.3.1]{JoP}.

\subsection{Dictionary}\label{sec:dictionary}

Our hope is that this paper will serve as a dictionary, whereby results from category theory literature can be imported directly into database theory. In Section \ref{sec:data migration} we will see such a result: translations between schemas provide data migration functors that have useful and provable properties. In Table \ref{table:dictionary} we gather some of the foundational links between databases and categories, as presented throughout the paper.

\begin{table}\caption{Dictionary between database terminology and category theory terminology.}\label{table:dictionary}\begin{tabular}{| p{2.42in} | p{2.1in} |}\bhline\multicolumn{2}{| c |}{Dictionary between DB and CT terminology}\\\bhline 
{\bf Database concept}&{\bf Category-theory concept}\\\bbhline 
Database schema, $\mcC$&Category, $\mcC$\\\hline 
Table $T\in\mcC$&Object $T\in\Ob(\mcC)$\\\hline 
Column $f$ of $T$&Outgoing morphism, $f\taking T\to\; ?$\\\hline 
Foreign key column $f$ of $T$ pointing to $U$&Morphism $f\taking T\to U$\\\hline 
Sequence of foreign keys&Composition of morphisms\\\hline 
Primary key column ID of $T$&Identity arrow, $\id_T\taking T\to T$\\\hline 
Controlled vocabulary (i.e. one-column table $D$) &Object $D$ without outgoing morphisms (except $\id_D\taking D\to D$)\\\hline 
Foreign key path equivalence in $\mcC$&Commutative diagram in $\mcC$\\\bhline 
Instance $I$ on schema $\mcC$&Functor $I\taking \mcC\to\Set$ (or  $I\taking\mcC\to\SS$ for some other ``nice" category $\SS$)\\\hline
Conversion of Relational to RDF&Grothendieck construction\\\hline 
Insertion update $u\taking I_0\to I_1$&Natural transformation $u\taking I_0\to I_1$\\\hline 
Deletion update $u\taking I_0\to I_1$&Natural transformation $u\taking I_1\to I_0$\\\bhline 
Schema mapping $\mcC\to\mcD$&Functor $F\taking\mcD\to\mcC$\\\hline 
Basic ETL process $\mcC\inst\to\mcD\inst$&Pullback functor, often denoted $\Delta_F$ or $F^*\taking\mcC\set\to\mcD\set$\\\bhline\end{tabular}\end{table}

\section{Data migration functors}\label{sec:data migration}

Given schemas $\mcC$ and $\mcD$, a data migration functor is tasked with transforming any $\mcC$-instance $I$ into some $\mcD$-instance $J$ (or vice versa). Moreover, it must do so {\em in a natural way}, meaning that progressive (resp. regressive) updates on $I$ must result in progressive (resp. regressive) updates on $J$. Data migration functors were concretely exemplified in Section \ref{sec:simplified data migration}.

In this section we will start with a translation between schemas $F\taking\mcC\to\mcD$. Recall from Definition \ref{def:translation} that this is simply a mapping from vertices in $\mcC$ to vertices in $\mcD$ and arrows in $\mcC$ to paths in $\mcD$, respecting path equivalence. Any translation $F$ generates three data migration functors. These will be denoted $$\Sigma_F\taking\mcC\inst\to\mcD\inst\hsp\Delta_F\taking\mcD\inst\to\mcC\inst\hsp\Pi_F\taking\mcC\inst\to\mcD\inst.$$ One may notice that $\Delta_F$ seems to go backwards---the direction opposite to that of $F$. Although this may seem counter-intuitive, in fact $\Delta_F$ is the simplest of the three data migration functors and the most straightforward to describe. 

Before we do so, let us quickly discuss value categories. Recall from Definition \ref{def:category of instances} that for any category $\SS$, we have a category $\mcC\inst_\SS$ of $\mcC$-instances valued in $\SS$, i.e. the category of functors $\mcC\to\SS$. The migration functor $\Delta_F\taking\mcD\inst_\SS\to\mcC\inst_\SS$ exists regardless of ones choice of $\SS$. For $\Sigma_F$ to exist, $\SS$ must be cocomplete, and for $\Pi_\SS$ to exist, $\SS$ must be complete. To fix ideas, most readers should simply take $\SS=\Set$, unless they are compelled to do otherwise. This is the case where the rows of each table form a set.  

\subsection{The pullback data migration functor $\Delta$}\label{sec:pullback}

Suppose we have a translation $F\taking\mcC\to\mcD$. Given a $\mcD$-instance, $I\in\mcD\inst_\SS$, we need to transform it to a $\mcC$-instance in a natural way. But this is simple, because $I\taking\mcD\to\SS$ is a functor and the composition of functors is a functor, so the composite $$\mcC\To{F}\mcD\To{I}\SS,$$ is an object of $\mcC\inst_\SS$, as desired. Similarly, a natural transformation $m\taking I\to J$ is whiskered with $F$ to yield a natural transformation $(m\circ F)\taking (I\circ F)\too (J\circ F)$. Thus we have defined a functor $$\Delta_F\taking\mcD\inst_\SS\to\mcC\inst_\SS,\hsp\Delta_F(-):=(-\circ F)$$ The slogan is ``$\Delta_F$ is given by composition with $F$." 

We have now defined the pullback functor $\Delta_F\taking\mcD\inst_\SS\to\mcC\inst_\SS$. It was explicitly discussed in Example \ref{ex:pullback}. Roughly, it can accommodate: renaming tables, renaming columns, deleting tables, projecting out columns, duplicating tables, and duplicating columns.

\subsection{The right pushforward data migration functor $\Pi$}\label{sec:right pushforward}

Suppose we have a translation $F\taking\mcC\to\mcD$. Given a $\mcC$-instance, $I\in\mcC\inst_\SS$, we need to transform it to a $\mcD$-instance in a natural way. We will do so by using the right adjoint of $\Delta_F$; however, to do this we will need to assume that $\SS$ is complete (i.e. that $\SS$ has small limits). Note that $\SS=\Set$ is complete.

\begin{proposition}

Let $F\taking\mcC\to\mcD$ be a functor, and let $\SS$ be a complete category. Then the functor $\Delta_F\taking\mcD\inst_\SS\to\mcC\inst_\SS$ has a right adjoint, which we denote by $$\Pi_F\taking\mcC\inst_\SS\to\mcD\inst_\SS.$$

\end{proposition}

\begin{proof}

This is \cite[Corollary X.3.2]{Mac}.

\end{proof}

We have now defined the right pushforward functor $\Pi_F\taking\mcC\inst_\SS\to\mcD\inst_\SS$. It was explicitly discussed in Example \ref{ex:right pushforward}. Roughly, it can accommodate: renaming tables, renaming columns, and joining tables. To see this, one applies the ``pointwise formula" for right Kan extensions, e.g. as given in \cite[Theorem X.3.1]{Mac}; a more explicit formulation is given in \cite{Sp3}.

\subsection{The left-pushforward data migration functor $\Sigma$}\label{sec:left pushforward}

Suppose we have a translation $F\taking\mcC\to\mcD$. Given a $\mcC$-instance, $I\in\mcC\inst_\SS$, we need to transform it to a $\mcD$-instance in a natural way. We will do so by using the left adjoint of $\Delta_F$; however, to do this we will need to assume that $\SS$ is cocomplete  (i.e.that $\SS$ has small colimits). Note that $\SS=\Set$ is cocomplete.

\begin{proposition}

Let $F\taking\mcC\to\mcD$ be a functor, and let $\SS$ be a cocomplete category. Then the functor $\Delta_F\taking\mcD\inst_\SS\to\mcC\inst_\SS$ has a left adjoint, which we denote by $$\Sigma_F\taking\mcC\inst_\SS\to\mcD\inst_\SS.$$

\end{proposition}

\begin{proof}

This follows from \cite[Theorem 3.7.2]{Bor1} and \cite[Proposition 9.11]{BW2}.

\end{proof}

We have now defined the left pushforward functor $\Sigma_F\taking\mcC\inst_\SS\to\mcD\inst_\SS$. It was explicitly discussed in Example \ref{ex:left pushforward}. Roughly, it can accommodate: renaming tables, renaming columns, taking the union of tables, and creating Skolem variables. To see this, one applies the ``pointwise formula" for left Kan extensions, e.g. as given in \cite[Theorem 3.7.2]{Bor1}; a more explicit formulation is given in \cite{Sp3}.

\section{Data types and filtering}\label{sec:typing}

In this section we will formulate the typing relationship that holds between abstract data and its representation. Until now, we have been considering data as simply a collection of interconnected elements---an instance in the sense of Definition \ref{def:instance} keeps track of various sets of abstract elements, segregated into tables and connected together in precise ways. However, in reality, each such element is represented in its table by way of a datatype, such as strings or integers. Semantically, each cell in a given column $c\taking t\to t'$ of a table $t$ should have the same datatype, namely they should all have the datatype of the target table, $t'$.

However, datatypes not only give a uniform method for displaying each data element, but they can also carry a notion of value. For example, salaries are numbers that can be added together to give meaningful invariants. It is for this reason that we must find a connection between database formalism and programming language formalism, as was discussed in Section \ref{sec:data and program}. Category theory provides such a connection: both database schemas and programming languages form categories, and these categories can be related by functors. 

Below we will explain these concepts, in particular how to attach datatypes from a programming language to tables in a categorical schema. To do so, we will make use of the concepts in Section \ref{sec:data migration}. After defining type signatures on schemas we will proceed to define morphisms of type signatures, which will enable us to filter data. For example, to filter all names that start with the letter R, we might pull back along an inclusion $\{\text{R}\}\to\Str$, where $\Str$ is the set of strings.

\subsection{Assigning data types via natural transformation}\label{sec:DT as NT}

We begin with two examples to motivate the definition. 

\begin{example}\label{ex:set of integers}

What is meant mathematically by the phrase ``a set of integers"? Consider that it is a set labeled $X$, together with a function $f\taking X\to\ZZ$. Allowing our set of integers to change, we get different sets labeled $X$ and different functions labeled $f$, but the set $\ZZ$ of integers is unchanged. From the categorical perspective we can understand ``a set $X$ of integers" in a couple different ways:
\begin{enumerate}\item as a database instance $I\taking\mcC\to\Set$ on the schema $$\mcC:=\fbox{\xymatrix@1{\LTO{X}\ar[r]^f&\LMO{\ZZ}}},$$ such that the image on $\LMO{\ZZ}$ is fixed as $I(\LMO{\ZZ})=\ZZ$; or 
\item as a database instance $J\taking\mcD\to\Set$ on the schema $$\mcD:=\fbox{$\LTO{X}$}$$ equipped with a natural transformation $f\taking J\to\{\ZZ\}$ (where $\{\ZZ\}$ is shorthand for the functor $\mcD\to\Set$ given by $D(\LTO{X})=\ZZ$).
\end{enumerate}

\end{example}

\begin{example}\label{ex:times 50}

Suppose we have a database and that we would like to enforce a mathematical relationship between two columns, say $t$ and $d$, in a certain table {\tt X}. For example, it might be that column $t$, say ``time spent" (in an integer number of hours) is related to column $d$, say ``debt owed" (as a dollar-figure) by a mathematical function $d=r(t)$, say $$d(x)=r(t(x)):=\$50*t(x)$$ for each $x\in X$. Just as in Example \ref{ex:set of integers}, this situation could be categorically represented in a couple ways, but we will focus on only one. Namely, we understand it as the collection of database instances $J\taking\mcD\to\Set$ on schema $\mcD$ as exemplified below \begin{align}\label{dia:times 50}\mcD:=\parbox{.75in}{\fbox{\parbox{.75in}{\begin{center}$\ul{d = t   r}$\end{center}\xymatrix{\LTO{X}\ar[r]^t\ar[dr]_-d&\LTO{Y}\ar[d]^-r\\&\LTO{Z}}}}}\hspace{.5in}J({\tt X}):=\begin{tabular}{| l || l | l |}\bhline\multicolumn{3}{|c|}{X}\\\bhline {\bf ID}&{\bf t}&{\bf d}\\\bbhline CtrX13&4&\$200\\\hline CtrX14&7&\$350\\\hline CtrX15&2&\$100\\\bhline\end{tabular}\end{align} The fact that the $d$ has dollar-figure datatype and that $d=\$50*t$ are enforced by a certain natural transformation, $f\taking J\to P$, where $P$ is a {\em typing instance}. We will give more details in Example \ref{ex:times 50 detail}.

\end{example}

\begin{definition}

Let $\mcC$ be a schema and let $P\in\Ob(\mcC\inst)$ be an instance. The category of {\em $P$-typed instances on $\mcC$}, denoted $\mcC\inst_{/P}$, is defined to be the ``slice" category of instances over $P$ (see, e.g. \cite[Categorical Preliminaries]{MM}). In other words, a $P$-typed instance on $\mcC$ is a pair $(I,\tau)$ where $I$ is an instance and $\tau\taking I\to P$ is a natural transformation; and a morphism of $P$-typed instances is a commutative triangle. 

\end{definition}

\begin{remark}

Given a schema $\mcC$ we may refer to any instance $P\in\Ob(\mcC\inst)$ as a {\em typing instance} if our plan is to consider $P$-typed instances, i.e. the category $\mcC\inst_{/P}$.

\end{remark}

\begin{remark}

Fix a schema $\mcC$ and a category $\SS$ and let $\mcE:=\mcC\inst_\SS$ denote the category of $\SS$-valued instances on $\mcC$. In practice $\mcE$ is often a topos (see Proposition \ref{prop:functor categories}). In case it is, then for any instance $P\in\Ob(\mcE)$, the category $\mcE_{/P}$ of $P$-typed $\SS$-valued instances on $\mcC$ is again a topos (see \cite[Theorem IV.7.1]{MM}). 

\end{remark}

\begin{construction}\label{const:typing}

Suppose given a category $\Type$ of types for some programming language and an $\SS$-valued functor $V\taking\Type\to\SS$ which sends each type to its set (or $\SS$-object) of values. We often wish to use a fragment of $\Type$ to add typing information to our database schema $\mcC$. If the fragment is given by the functor $\mcB\To{F}\Type$ and $\mcB$ is associated to the schema via a functor $\mcB\To{G}\mcC$, $$\SS\Froom{V}\Type\Froom{F}\mcB\Too{G}\mcC,$$ then $P:=\Pi_\mcG\circ\Delta_\mcF(V)$ is the implied typing instance. We call the sequence $(\mcB,F,G)$ {\em the typing auxiliary} in this setup.

\end{construction}

\begin{example}\label{ex:times 50 detail}

We return to Example \ref{ex:times 50} with the language from Construction \ref{const:typing}. We will now describe a typing auxiliary. Let $\mcB$ be one-arrow category drawn below, and let $G\taking\mcB\to\mcD$ be the suggested functor $$\mcB:=\parbox{.4in}{\fbox{\xymatrix{\LTO{Y'}\ar[d]^-{r'}\\\LTO{Z'}}}}\Too{G}\parbox{.9in}{\fbox{\parbox{.75in}{\begin{center}$\ul{d = t   r}$\end{center}\xymatrix{\LTO{X}\ar[r]^t\ar[dr]_-d&\LTO{Y}\ar[d]^-r\\&\LTO{Z}}}}}=:\mcD$$ Consider also the functor $F\taking\mcB\to\Type$ sending {\tt Y}$'$ to {\tt Int}, the type of integers, {\tt Z}$'$ to ${\tt Dollar}$ the type of dollar figures, and $r'$ to the function that multiplies an integer by 50. 

With $V\taking\Type\to\SS$ as in Construction \ref{const:typing}, the implied typing instance $P:=\Pi_G\circ\Delta_F(V)\taking\mcD\to\SS$ has $$P({\tt X})={\tt Int}\cross{\tt Dol};\hsp P({\tt Y})={\tt Int};\hsp P({\tt Z})={\tt Dollar},$$ and  $P(r)\taking P({\tt Y})\to P({\tt Z})$ is indeed the multiplication by 50 map.

Now a $P$-typed instance $\tau\taking I\to P$ is exactly what we want. For each of {\tt X, Y, Z} it is a set with a map to the given data type, and the naturality of $\tau$ ensures the properties described in Example \ref{ex:times 50} (i.e. that $d=\$50*t$ in the table $J({\tt X})$ from Display (\ref{dia:times 50})). In other words, it ensures that for any row in $J({\tt X})$, the value of the cell in column $d$ will be 50 times the value of the cell in column $t$.

\end{example}

\subsection{Morphisms of type signatures}\label{sec:morphism of types}

Each data migration functor discussed in Section \ref{sec:data migration} is a kind of tool for schema evolution. As new tables and columns are created and others are discarded, the translation between old schema and new will induce data migration functors that convert seamlessly from old data to new (and vice versa) and from queries against the old schema to queries against the new one (see also \cite{Sp4}). 

There is a slightly different kind of schema evolution that comes up often, namely changing data types. For example, if a company surpasses around 32,000 employees, they may need to change the datatype on their {\tt Employee} table from a smallint to a bigint. More complex changes include cutting the price for every share of stock by half, or concatenating a first and a last name pair to form a new field. Importantly, one needs to be able to reason about how queries against today's schema will differ from those against yesterday's. Change-of-types functors are just like data migration functors, and the formal nature of their description allows one to reason about their behavior.

Let $\mcC$ be a schema and let $\mcE$ be the topos $\mcC\inst$. Given a morphism of typing instances $k\taking P\to Q$, there are induced adjunctions $$\Adjoint{\wh{\Sigma}_k}{\mcE_{/P}}{\mcE_{/Q}}{\wh{\Delta}_k}\hspace{.5in}\Adjoint{\wh{\Delta}_k}{\mcE_{/Q}}{\mcE_{/P}}{\wh{\Pi}_k}$$ In other words, $k$ induces an essential geometric morphism of toposes (\cite[Theorem IV.7.2]{MM}).

\begin{definition}

Let $\mcC$ be a schema and $k\taking P\to Q$ a morphism of typing instances. We refer to the induced functors $$\wh{\Sigma}_k,\wh{\Pi}_k\taking\mcC\inst_{/P}\to\mcC\inst_{/Q}\hsp\wh{\Delta}_k\taking\mcC\inst_{/Q}\to\mcC\inst_{/P}$$ as {\em type-change functors}. To be more specific, $\wh{\Sigma}_k$ will be called the {\em left pushforward type-change functor}, $\wh{\Pi}_k$ will be called the {\em right pushforward type-change functor}, and $\wh{\Delta}_k$ will be called the {\em pullback type-change functor}.

\end{definition}

\begin{example}

We return to Example \ref{ex:times 50} with $\mcD$ and $P$ as defined there. Consider the left pushforward type-change functor $\wh{\Sigma}_k$ in the case that $k\taking P\to Q$ sends a dollar figure $x$ to True if $x\geq \$200$ and False if $x<\$200$. The functor $\wh{\Sigma}_k$ converts the table on the left to the table on the right below:
$$\tau:=
\begin{tabular}{| l || l | l |}\bhline\multicolumn{3}{|c|}{X}\\\bhline {\bf ID}&{\bf t}&{\bf d}\\\bbhline CtrX13&4&\$200\\\hline CtrX14&7&\$350\\\hline CtrX15&2&\$100\\\bhline
\end{tabular}
\hspace{.5in}
\wh{\Sigma}_k(\tau)=\begin{tabular}{| l || l | l |}\bhline\multicolumn{3}{|c|}{X}\\\bhline {\bf ID}&{\bf t}&{\bf d}\\\bbhline CtrX13&4&True\\\hline CtrX14&7&True\\\hline CtrX15&2&False\\\bhline
\end{tabular}$$

The other type-change functors, $\wh{\Pi}_k$ and $\wh{\Delta}_k$ not have useful results in the context of this particular example, but see Examples \ref{ex:satisfaction} and \ref{ex:filtering}.

\end{example}

The right type-change functor $\wh{\Pi}$ handles what might be called ``group satisfaction." Suppose we have a bunch ($P$) of people and a set $(I)$ of items are distributed among them ($\tau\taking I\to P$). If the people are then subdivided into groups ($k\taking P\to Q$), then we can ask each group $q\in Q$, ``how many ways are there for each of your members to offer up one of their items?" (cardinality of $\wh{\Pi}_k(\tau)^\m1(q)\ss\wh{\Pi}_k(I)$). For example, if one of the people was handed an empty set of items then his or her group will have no such joint offering ($\wh{\Pi}_k(\tau)^\m1(q)=\emptyset$). We now explain this by example.

\begin{example}\label{ex:satisfaction}

Let $\SS$ be a complete category and write $\mcC\inst$ instead of $\mcC\inst_\SS$. In this example we explain  how, given a morphism of typing instances $k\taking P\to Q$ on a schema $\mcC$, the type-change functor $\wh{\Pi}_k\taking\mcC\inst_{/P}\to\mcC\inst_{/Q}$ operates on a $P$-typed instance to return a $Q$-typed instance by what we called {\em group satisfaction} above. Suppose we have $\mcC$ and $\mcB$ as drawn, $$\mcB=\fbox{$\LTO{L'}$},\hsp\mcC=\fbox{$\LTO{L}\Too{f}\LTO{M}$}$$ with the functor $G\taking\mcB\to\mcC$ given by ${\tt L'}\mapsto {\tt L}$ and the functor $P'\taking\mcB\to\SS$ given by $P'({\tt L'})=\{1,2,3,4\}$. Let $Q'\taking\mcB\to\SS$ be given by $Q'({\tt L'})=\{x,y\}$ and let $k'\taking P'\to Q'$ be the map sending $1,2\mapsto x; \;\;3,4\mapsto y.$ Finally, let $$P:=\Pi_G(P'),\;\;\; Q=\Pi_G(Q'), \text{ and }\; k:=\Pi_G(k')\taking P\to Q.$$ We are ready to compute the right type-change functor along $k$ on any $P$-typed instance $I\to P$; we just need to write down some such instance. So, if $I\in\mcC\inst$ is the table on the left then $\wh{\Pi}_k(I)$ is the table in the middle: $$I:=\footnotesize
\begin{tabular}{| l || l |}\bhline\multicolumn{2}{| c |}{L}\\\bhline {\bf ID}&{\bf f}\\\bbhline a&1\\\hline b&2\\\hline c&1\\\hline d&3\\\hline e&2\\\hline f&4\\\hline g&2\\\bhline
\end{tabular}\hspace{.6in}
\wh{\Pi}_k(I)=\begin{tabular}{| l || l |}\bhline\multicolumn{2}{| c |}{L}\\\bhline {\bf ID}&{\bf f}\\\bbhline(a,b)&x\\\hline (a,e)&x\\\hline (a,g)&x\\\hline (c,b)&x\\\hline (c,e)&x\\\hline (c,g)&x\\\hline (d,f)&y\\\bhline
\end{tabular}
\normalsize
\hspace{-1.6in}
\parbox{1in}{$$
\begin{tabular}{cccc}\\&b&&\\a&e&&\\c&g&d&f\\\hline 1&2&3&4\\
\end{tabular}
$$\vspace{-.25in}
$$
\underbrace{\hspace{.3in}}\;\;\underbrace{\hspace{.3in}}
$$
$$x\hsp y$$}
$$
and a sketch of the reasoning is given on the right.

\end{example}

\subsection{Filtering data}\label{sec:filtering}

In this section we show how to use pullback type-change functor $\wh{\Delta}$ to filter data (e.g. answer queries like ``return the set of employees whose salary is less than \$100." 

In fact, given a morphism of instances $k\taking P\to Q$, the associated pullback type-change functor $\wh{\Delta}_k$ can either filter or multiply data (or both) depending on the injectivity or surjectivity of $k$. In this short section we concentrate only on filtering, because it appears to be more useful in practice. We work entirely by example; the definition of $\wh{Delta}$ is given in Section \ref{sec:morphism of types} or in the literature (\cite[Section IV.7]{MM}).

\begin{example}\label{ex:filtering}

As advertised above, we show how to filter employees by their salaries, in particular showing only those with salaries less than \$100. 

Suppose we have a schema $\mcC$ and two typing auxiliaries, $(\mcB,P,G)$ and $(\mcB,P',G)$, shown below: $$\Type
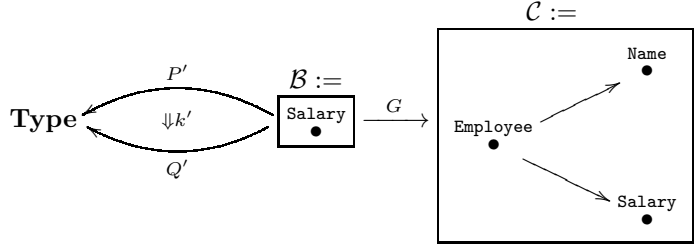
\hspace{-.2cm}\xymatrix{~&\hspace{.2in}&~\ar@/^1pc/[ll]^{Q'}\ar@/_1pc/[ll]_{P'}\ar@{}[ll]|{\Down k'}}\hspace{-.2cm}\parbox{.4in}{\vspace{-.15in}\begin{center}$\mcB:=$\end{center}\fbox{$\LTO{Salary}$}}\Too{G}\parbox{1.2in}{\begin{center}$\mcC:=$\end{center}\fbox{\xymatrix@=10pt{&&\LTO{Name}\\\LTO{Employee}\ar[rru]\ar[rrd]\\&&\LTO{Salary}}}}$$ Here we want $Q'({\tt Salary})$ to be the dollar-figure data type, we want $P'({\tt Salary})$ to be the subtype given by requiring that a dollar figure $x$ be less than \$100, and we want $k'\taking P'\to Q'$ to be the inclusion. These typing auxiliaries induce a morphism of typing instances $$k:=\Pi_G(k')\taking P\to Q,\;\text{ where }\; P:=\Pi_G(P')\;\text{ and }\;Q:=\Pi_G(Q').$$

Now suppose that $I\in\mcC\inst$ is the $Q$-typed $\mcC$-instance shown to the left below. Then $\wh{\Pi}_k(I)$ is the $P$-typed instance to the right below. 
$$\small I:=\begin{tabular}{| l || l | l |}\bhline
\multicolumn{3}{| c |}{\tt Employee}\\\bhline
{\bf ID}&{\bf Name}&{\bf Salary}\\\bbhline
Em101&Smith&\$65\\\hline Em102&Juarez&\$120\\\hline Em103&Jones&\$105\\\hline Em104&Lee&\$90\\\hline Em105&Carlsson&\$80
\\\bhline
\end{tabular}
\hspace{.4in}
\wh{\Pi}_k(I)=\begin{tabular}{| l || l | l |}\bhline
\multicolumn{3}{| c |}{\tt Employee}\\\bbhline
{\bf ID}&{\bf Name}&{\bf Salary}\\\bbhline
Em101&Smith&\$65\\\hline Em104&Lee&\$90\\\hline Em105&Carlsson&\$80
\\\bhline
\end{tabular}
$$\normalsize

\end{example}

Thus we see that filtering is simply an application of the same data migration functor story.

\subsection{A normal form for data migration}

We have seen several different forms of data migration functors throughout this paper. One may perform a sequence of data migration functors, e.g. moving data from one schema to another, then changing the data types, and finally filtering the result. In many cases, such combination of data migration functors can be rewritten as a sequence of three: a pullback, a right pushforward, and a left pushforward; see \cite{SW}.
\footnote{An earlier version of this paper included a false claim here. It stemmed from the assumption that the results of \cite{GK} could be translated to this context and that therefore any sequence of data migration functors could be rewritten as a sequence of only three. This does not seem to be the case.}
This means that there is a normal form for a quite general class of queries. Any combination of such queries can be written in the form $\Delta_F\Pi_G\Sigma_H$ for some $F,G,H$. This form may not be optimal in terms of speed, but it can serve as a single input format for query optimizers.

\bibliographystyle{amsalpha}

\end{document}